%% file: Main.tex
\documentclass[11pt, a4paper, oneside, reqno]{amsart}

\input{style}

\input{command}

\title[Uncertain MASs with Optimization Missions and Event-Triggered Communications]{Uncertain Multi-Agent Systems with Distributed Constrained Optimization Missions and Event-Triggered Communications: Application to Resource Allocation}

\author[M. S. Sarafraz]{Mohammad Saeed Sarafraz}
\author[M. S. Tavazoei]{Mohammad Saleh Tavazoei}

\thanks{The authors are with the Electrical Engineering Department, Sharif University of Technology, Iran ({\tt \{Sarafraz,Tavazoei\}@ee.sharif.edu}).}

\begin{document}
\maketitle

\begin{abstract}
    This paper deals with solving distributed optimization problems with equality constraints by a class of uncertain nonlinear heterogeneous dynamic multi-agent systems. It is assumed that each agent with an uncertain dynamic model has limited information about the main problem and limited access to the information of \textcolor{black}{the state variables of the other agents}. A distributed algorithm that guarantees cooperatively solving \textcolor{black}{of} the constrained optimization problem by the agents is proposed. Via applying this algorithm, the agents do not need to continuously broadcast their data. It is shown that the proposed algorithm can be useful in solving resource allocation problems.
\end{abstract}
	\section{Introduction}\label{sec.introduction}
	{N}{ew} generations of \textcolor{black}{the} networked systems are becoming more considered in modern engineering applications. In these applications, various research subjects such as distributed optimization, distributed control, event-triggered implementation, and real-time control are raised alongside the issue of networked systems. One of the main issues posed in the field of networked systems is the control of multi-agent systems over communication networks. In some of these systems, the agents are in competitive environments~\cite{tajeddini2018mean}, whereas in some other ones they cooperatively try to reach an intended objective~\cite{minciardi2011optimal}. Furthermore, in many cases the control objective in control of a multi-agent system can be expressed in terms of solving an optimization problem. For example, in multi-agent systems the issues of consensus (static or dynamic), rendezvous, formation control, deployment, and resource allocation can be formulated as \textcolor{black}{some} optimization problems. In this framework, through a static/dynamic model each agent updates its decision\textcolor{black}{/state} variables which participate in the global optimization problem. Due to distributed structure of the objective function and also unknown local parameters/functions, agents need to exchange some information between themselves. \textcolor[rgb]{0,0,0}{The communication constraints, such as cost and blackout constraints, force us to use non-continuous data transmission methods, i.e. periodic discrete transmission algorithms and event-triggered methods, instead of continuous transmission of data between the agents. }
	
	\color{black}
	The problem investigated in this paper lies at the interface of these three main topics: (i) distributed convex constrained optimization, (ii) robust control of uncertain multi-agent systems, and (iii) event-triggered communications. More specifically, the aim is to propose a distributed algorithm with continuous-time computations alongside event-triggered communications among the neighbors to control an uncertain, possibly nonlinear, dynamic multi-agent system such that this system can solve a convex constrained optimization problem.
	\color{black}
	\subsection{\label{sec:Literature}Literature Review}
	The event-triggered control \textcolor{black}{strategy}, as an approach for implementing aperiodic control algorithms and also as an alternative for the typical periodic sampled-data control \textcolor{black}{methods} \cite{heemels2012introduction,wang2011event}, has been discussed in some early classic works \cite{bekey1966sensitivity}. Recent advances on this topic can be found in \textcolor[rgb]{0,0,0}{various} recently published \textcolor{black}{papers (For example, \cite{duan2018asynchronous} and \cite{liu2019event1}, which deal with the event-triggered control in discrete-time and continuous-time plants, respectively)}. Also, in recent years, several studies have been done to apply event-triggered control methods in distributed and networked systems \cite{peng2020distributed}. In this regard, there are \textcolor[rgb]{0,0,0}{various research} works concerning stability \textcolor{black}{analysis} \cite{zhang2019event} and consensus \cite{wang2020hybrid} in event-triggered control based distributed systems. Also, there are some works on proposing event-triggered optimization algorithms for multi-agent systems \cite{kia2015distributed}. Within the framework of event-triggered control, the major difference between the problems of stabilization and optimization is that in the stabilization problem the equilibrium point is known, whereas in the optimization problem the goal is to find an \textcolor{black}{unknown} equilibrium point specifying the solution of the \textcolor[rgb]{0,0,0}{considered} optimization problem. In fact, distributed optimization has been introduced as a more practical alternative for centralized optimization \cite{margellos2018distributed,nedic2009distributed}. Although the consensus problem in multi-agent systems can be expressed in terms of a distributed optimization problem, there are some works addressing more general forms of distributed optimization problems in multi-agent systems \cite{savkin2019method}. \textcolor{black}{Benefiting from the} saddle-point/primal-dual dynamics is \textcolor{black}{one of} the most common \textcolor{black}{approaches for proposing distributed algorithms in order to solve} distributed optimization problems. As an example, \cite{feijer2010stability} has \textcolor{black}{focused on solving} distributed convex problems subject to some inequality constraints by using primal-dual gradient dynamics. Also, \cite{richert2016distributed} has introduced discontinuous saddle-point algorithms for \textcolor{black}{solving} distributed optimization problems.
	
	\subsection{\label{sec:Contribution}Statement of Contributions}
	\textcolor{black}{It was discussed that the control objectives in a multi-agent system can be described as optimization programs. Nevertheless, there are a few research works with the subject of distributed optimization problems which have assumed dynamic models for the agents solving those problems. These works have mainly considered the special case of the consensus problem. 
		For example, \cite{yu2017consensus} has studied consensus in multi-agent systems with the agents characterized by double-integrator dynamics. In addition, \cite{liu2018distributed,su2016distributed} have proposed consensus algorithms in the case of dealing with linear agents. In this paper we consider a general form for the agent dynamics and a general form of the optimization problem. Unlike the references cited above, in this study the dynamic model of the agents can be nonlinear and uncertain. As another relevant study, the recent research work \cite{zhang2019data} has considered the cooperative consensus problem in heterogeneous nonlinear multi-agent systems. The key difference between the present paper and the above-cited references is that the considered optimization program for the dynamic multi-agent system is more general, and the obtained results are not limited to the consensus problem. Also, design of distributed optimization algorithms for dynamic multi-agent systems has been addressed in a few research works. For example, \cite{wang2016distributed} has proposed an algorithm for solving a convex optimization problem in a multi-agent system with heterogeneous nonlinear agents. The communication between the agents has been assumed to be continuous in this study. Furthermore, \cite{jokic2009constrained} has suggested an approach for solving a constrained convex optimization problem with a dynamic single-agent system. 
		To generalize the previous works, the current research intends to add to the body of knowledge on distributed event-triggered algorithms by using a dynamic multi-agent system for solving an optimization problem with equality constraints. The agents in the dynamic multi-agent system are assumed to be described by uncertain nonlinear dynamic models in this study. In addition, each agent has its own cost function with the goal of determining the component corresponding to it in the optimal solution. Furthermore, an event-triggered strategy for exchanging between the agents is developed to reduce the communication costs.} 
	
	The main contributions of this paper can be listed as follows:
	\begin{itemize}
		\item{\bf Proposing a decentralized Zeno-free event-triggered algorithm to solve a distributed constrained convex optimization:}  We propose a decentralized mechanism in which the components of a constrained optimization problem can jointly solve this  problem without the need to access continuously to others' data. For this purpose, we provide a fully distributed event-triggered mechanism. Also, it is proved that due to a novel method, this optimization system under the event-trigger mechanism does not have Zeno-behavior (Theorem \ref{thm.decentralized}).

		\item{\bf Solving optimization problems by heterogeneous uncertain nonlinear multi-agent systems:} We develop a method for a distributed constrained optimization problem to be solved by a heterogeneous uncertain nonlinear multi-agent system. The involved agents use an event-based method for data exchanges (Theorem ~\ref{thm.main}). To the best of our knowledge, this is the first time that an uncertain nonlinear dynamic multi-agent system has been considered to solve this type of optimization problem.
	\end{itemize}
	\color{black}
	\subsection{\label{sec:notation}Notations \textcolor{black}{and Definitions}}
	
	In this paper, $\R$, $\R^n$, $\R^{m\times n}$, and $\N$ respectively denote the sets of real numbers, vectors with $n$ real elements, $m \times n$  real-valued matrices, and positive integer numbers. \textcolor{black}{The Euclidean inner product of vectors $x,y \in \R^n$ is denoted by $\left\langle {x,y} \right\rangle $. Also, the undirected graph $\mathcal{G}$ is described by $\mathcal{G}=(\mathcal{V}_{\mathcal{G}},\mathcal{E}_{\mathcal{G}})$ where $\mathcal{V}_{\mathcal{G}}$ is the vertex set and $\mathcal{E}_{\mathcal{G}} \subseteq \mathcal{V}_{\mathcal{G}} \times \mathcal{V}_{\mathcal{G}}$ is the edge set (Since graph $\mathcal{G}$ is undirected, $(i,j) \in \mathcal{E}_{\mathcal{G}}$ if and only if $(j,i) \in \mathcal{E}_{\mathcal{G}}$).} The symbol $\rho(A)$ specifies the maximum eigenvalue of real symmetric matrix $A$. In addition to the above-mentioned notations, \textcolor{black}{we use $\mathcal{D}(f)$ and $\mathcal{R}(f)$ to denote the domain and range of function $f: \mathcal{D}(f) \rightarrow \mathcal{R}(f)$, respectively.} \textcolor{black}{Moreover, the following definition is used in the next sections.}
	\begin{Def}(Monotonicity)
		\label{def.monotone}
		Function $g: \mathcal{D}(g)_{\subseteq \R^n} \rightarrow \R^n$ is monotonic if $\left\langle {x - y,g(x) - g(y)} \right\rangle  \ge 0$ for all $x,y \in \mathcal{D}(g)$. Also, $g$ is strictly monotonic if $\left\langle {x - y,g(x) - g(y)} \right\rangle  > 0$ for all non-equal $x,y \in \mathcal{D}(g)$.
	\end{Def}
	
	\subsection{\label{sec:Organization}Organization}
	
	The remainder of the paper is organized as follows. The \textcolor{black}{under study} problem is formulated with some basic assumptions in Section \ref{sec:problemstatement}. Main results of the paper are presented in Section \ref{sec:mainresults}. Section \ref{sec.application} deals with introducing sample applications \textcolor{black}{and some examples to illustrate the  theoretical results}. Finally, Section \ref{sec:conclusion} concludes the paper.
	
	\section{\label{sec:problemstatement} Problem Statement}
	
	Consider a network of $n$ agents ($n \in \N$) whose communication topology is described by the undirected and connected graph $\mathcal{G}=(\mathcal{V}_{\mathcal{G}},\mathcal{E}_{\mathcal{G}})$.  \textcolor{black}{The agent $i$ ($i=1,...,n$) in this network is described by the \textcolor{black}{uncertain} nonlinear dynamic model}
	\begin{align}
		\label{dynamic}
		{\dot x_i} = \left( {{p_i}\left( {{x_i}} \right) + \Delta {p_i}\left( {{x_i}} \right)} \right) + \left( {{b_i} + \Delta {b_i}} \right){u_i},
	\end{align}
	\textcolor{black}{where ${x_i}(t) \in \R$, $b_i \in \R$ is a known parameter \textcolor{black}{(as the nominal value of the input gain)}, and ${p_i}: \R \rightarrow \R$ is a known function \textcolor{black}{(as the nominal dynamic function of the model)}. Furthermore, $\Delta b_i \in \R$ and $\Delta p_i: \R \rightarrow \R$ are an unknown constant and an unknown function, respectively. Assume that the agent $i$ in the above-mentioned network has the objective function ${f_i}:\mathbb{R} \to \mathbb{R}$ for $i=1,...,n$}. The aim is to control \textcolor{black}{the agents }of the multi-agent system in such a way that they cooperatively solve the optimization problem
	\begin{align}
		\label{optimization}
		\left\{\begin{matrix*}[l]
			\min\limits_{x \in \R^n} & f(x) = \sum\limits_{i = 1}^n {f_i}({x_i})  \\
			~{\rm s.t.}  & Cx = d
		\end{matrix*},\right.
	\end{align}
	where $x \coloneqq \begin{bmatrix}x_1&x_2&...&x_n\end{bmatrix}^{\top}$, $C \in \R^{m \times n}$ $(m \leq n)$ \textcolor{black}{is a full row-rank matrix}, and $d=\begin{bmatrix} d_1&d_2&...&d_m\end{bmatrix}^{\top} \in \R^m$.
	
	In this paper, the following assumptions \textcolor{black}{on multi-agent system \eqref{dynamic} are considered.}
	\color{black}
	\begin{As}
		{[Assumptions on the dynamic model of the agents]}
		\label{as.coeff}
		Consider the multi-agent system~\eqref{dynamic}. In this system,
		\begin{enumerate}
			\item (Lipschitz smoothness assumption)
			\label{as.lipschitz}
			the function $p_i(x_i)$ is smooth and Lipschitz with parameter $L_{p_i}$ for all $i=1,...,n$.
			\item (Boundedness assumption)
			\label{as.bounded}
			functions $\frac{{\partial }}{{\partial {x_i}}}{p_i}({x_i})$, $\Delta p_i(x_i)$, and $\frac{\partial }{{\partial {x_i}}}\Delta {p_i}({x_i})$ are bounded for all $i=1,...,n$. Furthermore, there is a known bound $\rho_b$ such that $|\Delta b_i| \leq \rho_b$ for all $i=1,...,n$.
		\end{enumerate}
	\end{As}
	\textcolor{black}{The uncertain nonlinear model~\eqref{dynamic} with Assumption~\ref{as.coeff} provides a rich dynamical framework emerging in a wide range of problems. For example, the unknown linear dynamic systems under the influence of sinusoidal disturbances considered in \cite{bacsturk2012adaptive,pyrkin2010rejection} are particular cases of~\eqref{dynamic} with Assumption~\ref{as.coeff}. \textcolor{black}{As another example, \cite{wang2017distributed} has considered a class of continuous-time multi-agent systems influenced by unknown-frequency disturbances. This class is a special form of ~\eqref{dynamic} meeting Assumption~\ref{as.coeff}. Moreover, an unknown linear time invariant multi-agent system with bounded uncertainties can be considered as a special case of ~\eqref{dynamic} which satisfies Assumption~\ref{as.coeff}.}}
	
	In addition to the assumptions on dynamic model of the multi-agent system~\eqref{dynamic}, in this paper we consider the following assumptions on the optimization problem \eqref{optimization}.
	
	\begin{As}%
		{[Assumptions on the optimization problem]}
		\label{as.optimization}
		Consider the constrained optimization problem~\eqref{optimization}. In this problem, 
		\color{black}
		\begin{enumerate}
			\item (\textcolor[rgb]{0,0,0}{Assumptions on the convexity of the objective function and privacy})
			\label{as.convex}
			the function $f_i(x_i)$ is strictly convex and twice differentiable with bounded and continuous second derivative for all $i=1,...,n$. Also, the function $f_i(x_i)$ is only known for the agent $i$, and the other agents are not aware of it.
			\item (\textcolor[rgb]{0,0,0}{Assumption on the compatibility between the constraints and communication graph})
			\textcolor{black}{in the optimization problem~\eqref{optimization}, consider matrix $C$ as $C = \left[ c_{ij} \right]_{m \times n} = \begin{bmatrix}C_1^\top & \dotsc&C_m^\top\end{bmatrix}^\top \in \R^{m \times n}$ where the constraints in this optimization problem are $\left\{C_i^\top x-d_i = 0\right\}_{1 \leq i \leq m}$. It is assumed that matrix $C$ is compatible with the communication graph $\mathcal{G}=(\mathcal{V}_{\mathcal{G}},\mathcal{E}_{\mathcal{G}})$. This assumption means that if the state variables of two agents $j$ and $k$ are appeared in the $i^{th}$ constraint (or equivalently \textcolor{black}{$c_{ij} \times c_{ik} \neq 0$}), then $(j,k) \in \mathcal{V}_{\mathcal{G}}$.} \textcolor{black}{ Furthermore, it is assumed that each agent has access to the information of the state variables of its neighbors and the constraints involving such an agent (i.e., the constraint equation(s) in which the state variable of the agent is involved).}
			\item (Assumption on the spectral radius of matrix $C$)
			without loss of generality, it is assumed that $\rho(C^{\top}C) = 1$ (If $\rho(C^{\top}C) = p$, substituting the equality constraint of the optimization problem \eqref{optimization} by $C^\prime x=d^\prime$, where $C^\prime=C/ \sqrt{p}$ and $d^\prime=d/ \sqrt{p}$, satisfies this assumption).   
		\end{enumerate}
	\end{As}
	{It is worth noting that there are some widely used classes of cost functions satisfying the boundedness condition in Assumption~\ref{as.optimization}-\ref{as.convex}, e.g. the quadratic cost functions and the cost functions with Lipschitz continuous gradients. Some research works such as \cite{yi2016initialization,ghadimi2013multi} have considered these forms of cost functions in their studies.}
	\newline
	\color{black}
	In addition to the above assumptions, assume that due to some network constraints and also communication and computation costs, agents cannot continuously share their state \textcolor{black}{variables} with their neighbors. \textcolor{black}{To deal with this constraint, there are two main approaches: (i) periodic communications and (ii) event-triggered mechanism based communications. It is well known that due to avoid unnecessary communications,	the second approach, i.e. using the event-triggered mechanisms, needs less communications in comparison with the former one. In this paper, we focus on the event-based approach to deal with this constraint}. In summary, the problem considered in the rest of the paper can be formulated as follows.
	\begin{Prob}
		\label{problem}
		Consider the multi-agent system \eqref{dynamic} under Assumption~\ref{as.coeff}. Also, assume that \textcolor{black}{the} agents of this system aim to cooperatively solve the optimization problem \eqref{optimization} which satisfies Assumption~\ref{as.optimization}. \textcolor{black}{To this aim, if $\mathcal{N}^i$ is the set of the indices of the agent $i$'s neighbors ($1 \leq i \leq n$), synthesize {$\{x_i\}_{i \in \{i,\mathcal{N}^i\}_{[0,t]}} \mapsto u_i(t)$} (as the control signal)} and a distributed triggering mechanism for communication between the agents such that $x(t)=\begin{bmatrix}x_1(t)&x_2(t)&...&x_n(t)\end{bmatrix}^{\top}$ converges to the the \textcolor{black}{solution} of the optimization problem \eqref{optimization}.
	\end{Prob}
	\section{\label{sec:mainresults} Main Results}
	
	The main results of the paper are presented in this section. 
	\color{black} At first, by extending the existing results in the literature, we obtain a distributed algorithm with discontinuous communications to solve a convex optimization with equality constraints. At the second step, benefiting from the proposed algorithm and some techniques in robust control, a solution for Problem~\ref{problem} is proposed. 
	\subsection{\label{subsec.event} Distributed optimization with event-triggered communications}
	
	In this section, a distributed technique for solving an optimization problem with centralized event-triggered communications is introduced. Then, this technique is improved in the viewpoint of using decentralized event-triggered mechanism, instead of the centralized one. Before presenting these results, we need to restate the
	\color{black}following lemma, which has been presented in \cite{cortes2018distributed}. 
	\begin{Lem} (Distributed optimization with equality constraints) \cite{cortes2018distributed}
		\label{lem.cortes}
		Consider the distributed constrained optimization problem
		\begin{align}
			\label{lem.opt}
			\left\{\begin{matrix*}[l]
				\min\limits_{y \in \R^n} & f(y) = \sum\limits_{i = 1}^n {f_i}({y_i})  \\
				~{\rm s.t.}  & Cy = d
			\end{matrix*},\right.
		\end{align}
		where $y_i \in \R$ for all $1\leq i\leq n$, $y \coloneqq\begin{bmatrix}y_1&y_2&...&y_n\end{bmatrix}^{\top}$, $C \in \R^{m \times n}$, and $d \in \R^m$. Also, it was assumed that $Rank(C)=m  \leq n$. Assume that $f$ is differentiable with locally Lipschitz partial derivatives. If $y^{\star} \coloneqq \begin{bmatrix} y_1^{\star}&y_2^{\star}&...&y_n^{\star} \end{bmatrix}^{\top} \in \R^n$ is the solution of problem \eqref{lem.opt} and $\nabla f$ is strictly monotonic, then the equilibrium point corresponding to $y=y^{\star}$  in dynamic system
		\begin{align}
			\label{lem.dynamical}
			\left \{\begin{matrix*}[l]
				\dot y(t) =  - \nabla f(y) - C^\top{(Cy - d)} - C^\top \mu,  \\
				\dot \mu (t) = Cy - d, 
			\end{matrix*} \right .,
		\end{align}
		is asymptotically stable.
	\end{Lem}
	\color{black}
	On the basis of Lemma~\ref{lem.cortes}, we provide a centralized event-triggered framework which can guarantee the stability of $y=y^\star$ in the dynamic system~\eqref{lem.dynamical}. \textcolor{black}{In comparison to the results of ~\cite{cortes2018distributed}, the improvements of the achievements of the present paper are twofold. Firstly, for interactions between agents, discontinuous communication is replaced by the continuous one in order to reduce the communication costs. Secondly, the results of this paper are obtained by assuming that the involved agents are described by uncertain nonlinear models, whereas no dynamic model has been considered in~\cite{cortes2018distributed} for the agents.} 
	In the centralized event-triggered mechanism, it is assumed that the information of each agent (i.e., the value of $y_i$ for agent $i$) is synchronously broadcast to its neighbors at times $\{t_k\}_{k \in \N}$. The broadcast time $t_{k}$ is specified when a triggering condition is satisfied. For the brevity of the exposition, the following notations 
	\begin{align}
		\label{hat_update}
		\begin{matrix*}[l]
			{{\hat y}}(t) = \begin{bmatrix} \hat{y}_1(t) & ... & \hat{y}_n(t) \end{bmatrix}^\top \coloneqq \begin{bmatrix} {y_1}({t_k})&...&  {y_n}({t_k})\end{bmatrix}^\top, \\
			{{\hat \mu}}(t) = \begin{bmatrix} \hat{\mu}_1(t) & ... & \hat{\mu}_m(t) \end{bmatrix}^\top \coloneqq \begin{bmatrix} {\mu_1}({t_k})&...&  {\mu_m}({t_k})\end{bmatrix}^\top,\\
			e_y(t) \coloneqq y(t)-y(t_k),\quad e_\mu(t) \coloneqq \mu(t)-\mu(t_k),
		\end{matrix*}
	\end{align}
	\textcolor{black}{for ${t_k} \leq t < {t_{k + 1}}$}, which specify some piecewise constant functions, are used in the rest of the paper. Now, we are ready to proceed with the centralized mechanism in the next theorem.
	\color{black}
	\begin{Thm}  (Optimization with centralized event-triggered communications)
		\label{thm.centralized}
		\textcolor{black}{Consider the optimization problem~\eqref{lem.opt} under Assumption~\ref{as.optimization}.}
		Suppose that the data is synchronously transferred between the neighbors at times $\{t_k\}_{k \in \N}$, and $y_i(t)$ and $\mu_l(t)$ are updated at $t \in [t_k,t_{k+1})$ by
		\begin{align}
			\label{central.dynamic}
			\left\{\begin{matrix*}[l]
				{\dot y_i}(t) = {g_i}\left( {\hat y(t),\hat \mu (t)} \right); & 1\leq i \leq n\\
				{\dot \mu }_l(t) = \sum\limits_{j = 1}^n {{c_{lj}}}{{\hat y}_j}(t) - d_l; & 1\leq l \leq m
			\end{matrix*}\right.,
		\end{align}
		where
		\begin{align}
			\nonumber
			{g_i}\left( {\hat y(t),\hat \mu (t)} \right) = 
			- \nabla {f_i}\left( {{{\hat y}_i}(t)} \right) - \sum\limits_{k = 1}^m {\left[ {\left( {\sum\limits_{j = 1}^n {{c_{kj}}{{\hat y}_j}(t) - {d_k}} } \right){c_{ki}}} \right] - \sum\limits_{l = 1}^m {{{\hat \mu }_l(t)}{c_{li}}} }, 
		\end{align}
		\textcolor{black}{and $\hat y(t)$ and $\hat \mu (t)$ are defined by \eqref{hat_update}}. Also, assume that $
		0 < \underset{\raise0.3em\hbox{$\smash{\scriptscriptstyle-}$}}{M} I \leqslant {\nabla ^2}f \leqslant \bar MI.$
		Let
		\begin{align}
			\label{eq.h_definition}
			h(\hat y(t),\hat \mu(t) ) = \left. {\left( {1 - 2\kappa } \right)\left\| {C\hat y(t) - d} \right\|_2^2} \right.
			+\left. {\left( {\underset{\raise0.3em\hbox{$\smash{\scriptscriptstyle-}$}}{M}  - \frac{{5 + 3 \bar M}}{2}\kappa } \right)\left\| {g\left( {\hat y(t),\hat \mu(t) } \right)} \right\|_2^2} \right.,   
		\end{align}
		where $g\left(\hat y(t),\hat \mu(t) \right) \coloneqq \begin{bmatrix}g_1(\hat y(t),\hat \mu(t))&...&g_n(\hat y(t),\hat \mu(t)) \end{bmatrix}^\top$ and
		\begin{align}
			\label{def.kappa}
			0<\kappa  < \min \left\{ {\dfrac{1}{{2}},\dfrac{{2}\underset{\raise0.3em\hbox{$\smash{\scriptscriptstyle-}$}}{M} }{{5+ 3\bar M}}} \right\}.
		\end{align}
		\textcolor{black}{S}pecify the event for triggering such that $\hat{y}_i(t)$ and $\hat{\mu_l}(t)$ are updated when condition
		\begin{align}
			\label{central.trigger}
			\left\{\begin{matrix*}[l]
				\frac{{2 + \bar M}}{\kappa }\left\| {{e_y(t)}} \right\|_2^2 + \frac{3}{{2\kappa }}\left\| {{e_\mu(t) }} \right\|_2^2 > h(\hat y(t),\hat \mu(t) ), \vspace{1mm}\\
				or \vspace{1mm}\\
				\left\{\begin{matrix*}[l]\frac{{2 + \bar M}}{\kappa }\left\| {{e_y(t)}} \right\|_2^2 + \frac{3}{{2\kappa }}\left\| {{e_\mu(t) }} \right\|_2^2 = h(\hat y(t),\hat \mu(t) ), \vspace{1mm}\\
					and \vspace{1mm}\\ \left\| {{e_y(t)}} \right\|_2^2 + \left\|{{e_\mu(t) }} \right\|_2^2 \neq 0,
				\end{matrix*}\right.
			\end{matrix*}\right.
		\end{align}
		holds. In this case, the equilibrium point corresponding to $y = {y^{\star}} = \begin{bmatrix} y_1^{\star}&...&y_n^{\star} \end{bmatrix}^{\top} \in {\mathbb{R}^n}$, as the solution of the optimization problem \eqref{lem.opt}, in dynamic system \eqref{central.dynamic} is asymptotically stable.
	\end{Thm}
	\begin{proof}
		Assume that $(y^{\star},\mu^{\star})$ is the equilibrium point of dynamic system \eqref{lem.dynamical}. \textcolor{black}{Firstly, we prove that $(y^{\star},\mu^{\star})$ is also an equilibrium point for system \eqref{central.dynamic}. If $(y^{\star},\mu^{\star})$ is the equilibrium point of dynamic system \eqref{lem.dynamical}, then we should have $Cy^\star=d$, or equivalently $C_l^\top y^\star = d_l$ for all $1 \leq l \leq m$. This means that $\sum\limits_{j = 1}^n {{c_{lj}}}{{y}_j^\star} - d_l = 0$ for all $1 \leq l \leq m$. Moreover, if $(y^{\star},\mu^{\star})$ is the equilibrium point of dynamic system \eqref{lem.dynamical}, then $-\nabla f(y^\star) - C^\top \mu^\star =0 $. Form this, it is deduced that $-\nabla f_i(y_i^\star) - \sum\limits_{l = 1}^m {{{ \mu }_l^\star}{c_{li}}} = 0$. According to equalities $\sum\limits_{j = 1}^n {{c_{lj}}}{{y}_j^\star} - d_l = 0$ for all $1 \leq l \leq m$ and $-\nabla f_i(y_i^\star) - \sum\limits_{l = 1}^m {{{ \mu }_l^\star}{c_{li}}} = 0$, it is found that $(y^{\star},\mu^{\star})$ is the equilibrium point of system~\eqref{central.dynamic}. To} prove the \textcolor{black}{asymptotic} stability of this equilibrium point in dynamic system \eqref{central.dynamic}, a Lyapunov function candidate is proposed and it is shown that its \textcolor{black}{time-}derivative is negative-definite with respect to dynamic model \eqref{central.dynamic}. Similar to that done in \cite{richert2016distributed}, the Lyapunov function candidate is chosen as $v(y,\mu ) = \sum\nolimits_{i = 1}^4 {{v_i}\left( {y,\mu } \right)} $, where
		\begin{align}
			\label{lyapunov}
			\left\{\begin{matrix*}[l]
				{{v_1}(y,\mu ) = \frac{1}{2}\left\| {g\left( {y,\mu } \right)} \right\|_2^2}, &
				{{v_2}(y,\mu ) = \frac{1}{2}\left\| {Cy - d} \right\|_2^2},\vspace{1mm}\\
				{{v_3}(y,\mu ) = \frac{1}{2}\left\| {y - {y^{\star}}} \right\|_2^2},&
				{{v_4}(y,\mu ) = \frac{1}{2}\left\| {\mu  - {\mu ^{\star}}} \right\|_2^2.}
			\end{matrix*}\right.
		\end{align}
		The derivative of \textcolor{black}{the} Lyapunov function candidate equals to 
		\begin{align}
			\nonumber
			\dot v(y,\mu ) = \sum\limits_{i = 1}^4 {\left( {\frac{\partial }{{\partial y}}{v_i}(y,\mu ) \cdot g\left( {\hat y,\hat \mu } \right) + \frac{\partial }{{\partial \mu }}{v_i}(y,\mu ) \cdot \left( {C\hat y - d} \right)} \right)}
		\end{align}
		By some computations, ${\dot v_i}(y,\mu )$ ($1\leq i \leq 4$) is obtained by
		\begin{align}
			\nonumber
			{{\dot v}_1}(y,\mu ) = & - g{\left( {\hat y,\hat \mu } \right)^{\top}}\left( {{\nabla ^2}f + {C^{\top}}C} \right)g\left( {\hat y,\hat \mu } \right)
			\\\nonumber &+ e_y^{\top}{C^{\top}}C\left( {C\hat y - d} \right) + e_\mu ^{\top}C\left( {C\hat y - d} \right)\\ \nonumber & +{\left( {\nabla f(y) - \nabla f(\hat y)} \right)^{\top}}\left( {{\nabla ^2}f + {C^{\top}}C} \right)g\left( {\hat y,\hat \mu } \right) \\ \nonumber & + {\left( {C\hat y - d} \right)^{\top}}Cg\left( {\hat y,\hat \mu } \right) \\ \nonumber
			&+e_y^{\top}{C^{\top}}C\left( {{\nabla ^2}f + {C^{\top}}C} \right)g\left( {\hat y,\hat \mu } \right) \\ \nonumber &+ e_\mu ^{\top}C\left( {{\nabla ^2}f + {C^{\top}}C} \right)g\left( {\hat y,\hat \mu } \right)
			\\ \nonumber&+ {\left( {\nabla f(y) - \nabla f(\hat y)} \right)^{\top}}{C^{\top}}\left( {C\hat y - d} \right), \\ \nonumber
			{\dot v_2}(y,\mu ) =  &- {\left( {C\hat y - d} \right)^{\top}}Cg\left( {\hat y,\hat \mu } \right) - e_y^{\top}{C^{\top}}Cg\left( {\hat y,\hat \mu } \right),\\
			\nonumber
			{{\dot v}_3}(y,\mu ) = & \,\, e_y^{\top}g\left( {\hat y,\hat \mu } \right) - {\left( {\hat y - {y^{\star}}} \right)^{\top}}\left( {\nabla f(\hat y) - \nabla f({y^{\star}})} \right)
			\\ \nonumber &- {\left( {\hat y - {y^{\star}}} \right)^{\top}}{C^{\top}}C\left( {\hat y - {y^{\star}}} \right) \\\nonumber&- {\left( {\hat y - {y^{\star}}} \right)^{\top}}{C^{\top}}\left( {\hat \mu  - {\mu ^{\star}}} \right),\\
			\nonumber
			{\dot v_4}(y,\mu ) = &\,\, e_\mu ^{\top}\left( {C\hat y - d} \right) + {\left( {\hat \mu  - {\mu ^{\star}}} \right)^{\top}}{C^{\top}}\left( {\hat y - {y^{\star}}} \right).
		\end{align}
		By using the Young inequality
		\begin{align}
			\nonumber
			{a^{\top}}b \leqslant \frac{\kappa }{2}{a^{\top}}a + \frac{1}{{2\kappa }}{b^{\top}}b;\qquad \forall a,b \in {\mathbb{R}^n},\,\kappa  > 0,
		\end{align}
		and the facts $0 \leqslant {x^{\top}}{C^{\top}}Cx \leqslant \|x\|^2$, and
		\begin{align}
			\nonumber
			\underset{\raise0.3em\hbox{$\smash{\scriptscriptstyle-}$}}{M} \left\| {x - y} \right\|_2^2 \leqslant {\left( {\nabla f(x) - \nabla f(y)} \right)^{\top}}\left( {x - y} \right) \leqslant \bar M\left\| {x - y} \right\|_2^2,
		\end{align}
		it can be easily found an upper bound for $\dot{v}(t)$ as follows.
		\begin{multline}
			\nonumber
			\dot v(t) \leqslant \frac{{2 + \bar M}}{\kappa }\left\| {{e_y}} \right\|_2^2 + \frac{3}{{2\kappa }}\left\| {{e_\mu }} \right\|_2^2- \left( {\underset{\raise0.3em\hbox{$\smash{\scriptscriptstyle-}$}}{M}  - \frac{{5 + 3\bar M}}{2}\kappa } \right)\left\| {g\left( {\hat y,\hat \mu } \right)} \right\|_2^2- \left( {1 - 2\kappa } \right)\left\| {C\hat y - d} \right\|_2^2 \\
			{\text{              }} = \frac{{2 + \bar M}}{\kappa }\left\| {{e_y}} \right\|_2^2 + \frac{3}{{2\kappa }}\left\| {{e_\mu }} \right\|_2^2 - h\left( {\hat y,\hat \mu } \right).
		\end{multline} 
		If parameter $\kappa$ meets condition \eqref{def.kappa} and $\hat{y}(t)$ and $\hat{\mu}(t)$ are updated when \eqref{central.trigger} holds, $\dot{v}(t)$ will be negative-definite. Hence, in this case the equilibrium point $(y^{\star},\mu^{\star})$  is asymptotically stable in dynamic system \eqref{central.dynamic}.
	\end{proof}
	\color{black}
	\begin{Rem}[Relaxation of the considered assumptions]
		Theorem~\ref{thm.centralized} was proved on the basis of Lemma~\ref{lem.cortes} which has been borrowed from~\cite{cortes2018distributed}. In the Reference~\cite{cortes2018distributed}, an augmented Lagrangian has been used to introduce the dynamic model~\eqref{lem.dynamical} and prove its convergence. The use of this augmented Lagrangian makes it necessary for agent $i$ ($1 \leq i \leq n$) to know the state information of all the other agents that are present with agent $i$ in at least one of the constraints in order to use dynamic~\eqref{lem.dynamical}. For this reason, in Assumption~\ref{as.optimization}, the compatibility between the matrix $C$ and graph $\mathcal{G}$  has been assumed. An idea to relax this assumption is to use the regular Lagrangian instead of the augmented one. In this case, the agents should use a dynamic model in the form
		\begin{align}
			\nonumber
			\left \{\begin{matrix*}[l]
				\dot y(t) =  - \nabla f(y) - C^\top \mu,  \\
				\dot \mu (t) = Cy - d, 
			\end{matrix*} \right .,
		\end{align}
		instead of the dynamic model~\eqref{lem.dynamical} to reach the convergence. As a consequence, the speed of the convergence will be naturally decreased in such a case. This idea invites further research works in the continuation of the study of this paper with the aim of relaxing the considered assumptions.
	\end{Rem}
	\color{black}
	Theorem~\ref{thm.centralized} introduced a method which can solve the optimization problem \eqref{lem.opt} by using a centralized {mechanism} with no need to continuous communications between the agents. The results of Theorem~\ref{thm.centralized} will be modified in Theorem~\ref{thm.decentralized} by introducing a decentralized version of the {mechanism}. In the modified method, each agent individually decides about the time of broadcasting its state \textcolor{black}{variable} information to neighbors according to the available local data (state \textcolor{black}{variables} information of its neighbors). The modified algorithm is proposed on the basis of the ideas presented in \cite{nowzari2016distributed}. Synchronous broadcasting at the situation \textcolor[rgb]{0,0,0}{that the broadcast times are not sufficiently far from each others} is the idea borrowed from paper \cite{nowzari2016distributed}. This idea is modified to introduce a distributed algorithm \textcolor[rgb]{0,0,0}{for solving} Problem~\ref{problem}.
	
	\color{black} Before introducing the decentralized version of the algorithm, we need to define the concept of virtual agents for the considered multi-agent system. Using this concept helps us to decentralize the communications in the under-study multi-agent system. \color{black} Assume that there are $m$ virtual agents in addition to $n$ primary (real) agents. Each of these virtual agents updates one of Lagrangian variables $\mu_l$ ($1 \leq l \leq m$). For simplicity, suppose that these virtual agents are indexed by $\{n+1,n+2,...,n+m\}$, and from here on, it is assumed that there are \textcolor{black}{$n+m$ agents ($n$ real agents and $m$ virtual ones) in the augmented multi-agent system}. Also, in this new augmented multi-agent system, we assume that each virtual agent has only primary (real) neighbors, and there are no two virtual neighbors. The primary (real) agent $i$ \textcolor[rgb]{0,0,0}{($i=1,...,n$)} is the neighbor of the virtual agent $n+j$ if and only if $c_{ij} \neq 0$, where $C=[c_{ij}]$ is the constraint matrix in the optimization problem~\eqref{lem.opt}. 
	\textcolor{black}{In a brief form, if $\mathcal{A}^{a} \in \R^{(m+n) \times (m+n)}$ is the adjacency matrix of the augmented communication graph (which is shown by $\mathcal{G}^a$), then
		\begin{align}
			\nonumber
			\mathcal{A}^a = \begin{bmatrix}
				\mathcal{A} & \mathcal{A}_{\mu} \\
				\mathcal{A}_{\mu}^\top & 0
			\end{bmatrix},
		\end{align}
		where $\mathcal{A}$ is the adjacency matrix of graph $\mathcal{G}$, and the $(i,j)$ element of $\mathcal{A}_{\mu}$ is equal to $1$ ($0$) if and only if $c_{ij} \neq 0$ ($c_{ij} = 0$).} Theorem~\ref{thm.decentralized} \textcolor{black}{can be applied to introduce} a decentralized algorithm for solving optimization problem \eqref{lem.opt} with discontinuous event-triggered communications.

	\begin{Thm} (Optimization with decentralized event-triggered communications)
		\label{thm.decentralized}
		\textcolor{black}{Consider the optimization problem~\eqref{lem.opt} under Assumptions~\ref{as.optimization} and $
			0 < \underset{\raise0.3em\hbox{$\smash{\scriptscriptstyle-}$}}{M} I \leqslant {\nabla ^2}f \leqslant \bar MI$, and the above-described augmented multi-agent system with $n+m$ agents ($n$ real agents with the state variables $y_i$ for $i=1,...,n$ and $m$ virtual ones with the state variables $\mu _{i-n}$ for $i=n+1,...,n+m$).} 
		\textcolor{black}{Assume that the agents update their state variables as}
		\textcolor{black}{
			\begin{align}
				\label{decentral.dynamic}
				\left\{\begin{matrix*}[l]
					{\dot y_i}(t) &= {g_i}\left( {\tilde y(t),\tilde \mu (t)} \right), & 1\leq i \leq n\\
					{\dot \mu }_{i-n}(t) &= \sum\limits_{j = 1}^n {{c_{(i-n)j}}}\,{{\tilde y}_j}(t) - d_{i-n}, & n+1\leq i \leq n+m
				\end{matrix*}\right.,
			\end{align}
			where
			\begin{align}
				\label{tilde.def}
				\begin{matrix*}[l]
					{{\tilde y}}(t) = \begin{bmatrix} \tilde{y}_1(t) & ... & \tilde{y}_n(t) \end{bmatrix}^\top = \begin{bmatrix} {y_1}({t^1_{k_1}})&...&  {y_n}({t^n_{k_n}})\end{bmatrix}^\top, \\
					{{\tilde \mu}}(t) = \begin{bmatrix} \tilde{\mu}_1(t) & ... & \tilde{\mu}_m(t) \end{bmatrix}^\top = \begin{bmatrix} {\mu_1}({t^{n+1}_{k_{n+1}}})&...&  {\mu_m}({t^{n+m}_{k_{n+m}}})\end{bmatrix}^\top
				\end{matrix*}
			\end{align}
			$t^i_{k_i} \in \{t^i_{k}\}_{k \in \N}$ is the last broadcast time of the information of the agent $i$ ($i=1,...,n+m$) before time $t$, and function $g_i$ is as that defined in Theorem~\ref{thm.centralized} \footnote{\color{black}The agent $i$ for using function $g_i$ needs to know only the subset 
				\begin{align}
					\nonumber
					\mathcal{C}_i=\left\{c_{ki} | 1 \leq k \leq n \right\} \cup \{c_{kj} | c_{kj}\cdot c_{ki} \neq 0\}
				\end{align}
				of the set of the elements of matrix $C$. Awareness of this information for the agent $i$ is consistent with Assumption II.2-2, in which it was assumed that matrix $C$ is compatible with the neighborhood graph $\mathcal{G}$.}.} Furthermore, assume that agent $i$ broadcasts its state \textcolor{black}{variable} to its neighbors if one of the following conditions
		\begin{align}
			\label{decental.trigger}
			\left\{\begin{matrix*}[l]
				(i) & 0 < {{r_i}(t) \leq r_{min_i},~~~~~~ 1 \leqslant i \leqslant n + m} \vspace{1.5mm} \\
				or \vspace{1mm}& \\
				(ii) &\left \{ \begin{matrix*}[l] \left\{\begin{matrix*}[l] {e^2_{{y_i}}(t)} > \gamma_i \left( {g_i^2\left( {\tilde y(t),\tilde \mu(t) } \right)} \right), \\ or \\ {e^2_{{y_i}}(t)} = \gamma_i \left( {g_i^2\left( {\tilde y(t),\tilde \mu(t) } \right)} \right)\, \&\, e^2_{y_i}(t) \neq 0,\end{matrix*} \right.,\\\qquad \qquad \qquad \qquad \qquad \qquad \qquad \qquad1 \leqslant i \leqslant n \vspace{1.5mm} \\ \left\{\begin{matrix*}[l] {e^2_{{\mu _{i - n}}}(t)} > \gamma_i \left( {c_{i - n}^{\top}\tilde y(t) - {d_{i - n}}} \right)^2, \\ or \\  {e^2_{{\mu _{i - n}}}(t)} = \gamma_i \left( {c_{i - n}^{\top}\tilde y(t) - {d_{i - n}}}\right)^2\, \&\, e^2_{\mu_i}(t) \neq 0, \end{matrix*}\right.,\\\qquad \qquad \qquad \qquad \qquad \qquad n + 1 \leqslant i \leqslant n + m \end{matrix*}\right.
			\end{matrix*}\right.
		\end{align}
		
		\color{black}
		occurs, where 
		\begin{align}
			\label{r_min_definition}
			{r_i}(t) \coloneqq \max \left\{ {\left. {t_{k_j}^j} \right|t_{k_j}^j < t \text{\space}\& \left( {i,j} \right) \in {\mathcal{E}_{\mathcal{G}^a}}} \right\} - \max \left\{ {\left. {t_{k_i}^i} \right|t_{k_i}^i < t} \right\},
		\end{align}
		\textcolor{black}{$e_{y_i}(t) \coloneqq y_i(t)-\tilde y_i(t),\,
			e_{\mu_{i-n}}(t) \coloneqq \mu_{i-n}(t)-\tilde \mu_{i-n}(t)$, and $r_{\min_i }$ and $\gamma_i$ are positive constants satisfying} 
		\begin{align}
			\label{gamma_definition}
			{{r_{\min_i }^2} < \gamma_i^2 } < min\left\{\dfrac{1}{12} , \dfrac{\underset{\raise0.3em\hbox{$\smash{\scriptscriptstyle-}$}}{M}^2}{2(5+3\bar{M})(\bar{M}+2)}\right\}; \quad 1\leq i \leq n+m.
		\end{align}
		In this situation, dynamic model \textcolor{black}{~\eqref{decentral.dynamic}} \textcolor[rgb]{0,0,0}{with the} broadcasting rule \eqref{decental.trigger} \textcolor[rgb]{0,0,0}{(distributed over the communication graph $\mathcal{G}^a$)} asymptotically converges to the solution of the optimization problem \eqref{lem.opt} without exhibiting Zeno behavior.
	\end{Thm}
	\textcolor{black}{The proof of this theorem can be found in Appendix. In this proof, Lyapunov functions similar to those proposed in Theorem~\ref{thm.centralized} are used to prove asymptotic convergence. Also, the proof of  the non-existence of the Zeno behavior is sketched based on proving the following claims.
		\newline
		\textit{Claim 1}: If the agent $i$ does not receive data from any of its neighbors, it will not broadcast the information of its state variable sooner than the time $t_{k_i - 1}^i + \gamma_i $, where $t_{k_i - 1}^i$ is the last broadcast time of the agent $i$.
		\newline
		\textit{Claim 2}: If all the agents synchronously broadcast their information at time $t_a$, then the next broadcast time of each agent cannot be less than $t_a + \underset{i}{\min}\, r_{min_i}$.
		\newline
		\textit{Claim 3}: If the described event-triggered mechanism exhibits the Zeno behavior, then for each $\epsilon$ satisfying 
		\begin{align}
			\label{newinequality1}
			0 < \epsilon < \dfrac{1}{n+m+1} \underset{i}{\min}\, r_{min_i},
		\end{align}
		there exists a time interval with a length smaller than $\underset{i}{\min}\, r_{min_i}$, such that not only the time between successive broadcasts of the information in this time interval is less than $\epsilon$, but also all the agents synchronously broadcast their information in one of these broadcast times.} \\
	\textcolor{black}{Considering the point that the asymptotic stability is equivalent to the exponential one in linear time invariant systems, \eqref{decentral.dynamic} exponentially converges to the solution of optimization problem~\eqref{lem.opt} in the cases that the $f_i(.)$ are quadratic cost functions.}
	\textcolor{black}{Generally speaking, the event-generating rules in event-triggered control methods can be classified into three general categories. In the first category, the error between the last communicated value of a specific signal and its actual value is compared to a fixed threshold. The event-generating rules obtained from this approach are called absolute event-triggered mechanisms \cite{zhou2019periodic}. The event-generating rules in the second category are constructed based on a comparison between the aforementioned error and a relative threshold which yields in relative event-triggered mechanisms~\cite{tabuada2007event}. In the third category, the mixed event-triggered mechanisms are obtained by the event-generating rules in which the error is compared with a threshold combined of fixed and relative terms~\cite{donkers2012output}. Using relative event-triggered mechanisms, the error can be converged to zero, whereas the absolute and mixed event-triggered mechanisms do not have such a feature. Nevertheless, the advantage of the absolute and mixed event-triggered mechanisms (in comparison to relative ones) is that if the error is close enough to zero, no more information will be exchanged, and no communication fee will be paid. The event-generating rule of Theorem \ref{thm.decentralized} generally follows a relative event-triggered mechanism. The condition $(ii)$ in relation~\eqref{decental.trigger} illustrates the relative comparison and attempts to update the  signals ($\{y_i\}_{1 \leq i \leq n},\{\mu_i\}_{n+1 \leq i \leq n+m}$) before the difference between the last sent value and the actual value of these signals becomes too large. On the other hand, one of the main challenges in introducing the event-triggered mechanisms is to guarantee the non-existence of Zeno behavior in the case of using such mechanisms. The condition $(i)$ in the event-generating rule~\eqref{decental.trigger} resolves the mentioned challenge. As seen in the proof of Theorem \ref{thm.decentralized}, this condition prevents the occurrence of the Zeno behavior by checking the broadcasting time of the agents.}	Not only Zeno-freeness is proved for the triggering mechanism~\eqref{decental.trigger}, but also the proof of Theorem~\ref{thm.decentralized} reveals that this mechanism yields in a dwell-time positivity property (On the basis of this proof, the time interval between each two consecutive broadcast times is not less than $\min\left\{r_{\min_i}\right\}_{i=1,...,n}$).
	
	The following remark reveals how the agents can deal with the global parameters which are used in~\eqref{gamma_definition}.
	\color{black}
	\begin{Rem} (Computing global parameters)
		\label{rem.global}
		According to \eqref{gamma_definition}, each agent needs to know global parameters $\bar{M}$ and $\underset{\raise0.3em\hbox{$\smash{\scriptscriptstyle-}$}}{M}$ for setting the \textcolor[rgb]{0,0,0}{local} parameters $\gamma_i$ and $r_{\min_i}$. A simple distributed algorithm can be proposed as follows \textcolor[rgb]{0,0,0}{in order} to help the agents for computing $\bar{M}$ and $\underset{\raise0.3em\hbox{$\smash{\scriptscriptstyle-}$}}{M}$. To this end, at the first step each agent should share \textcolor[rgb]{0,0,0}{the} upper and lower bounds of the second derivative of its cost function with its neighbors. Then, by considering the received information from the neighbors, each agent choose the lowest (greatest) value among the lower (upper) bounds reported by its neighbors \textcolor[rgb]{0,0,0}{to estimate} $\underset{\raise0.3em\hbox{$\smash{\scriptscriptstyle-}$}}{M}$ ($\bar{M}$). At the next step, each agent shares its estimations on the mentioned parameters with its neighbors, and then again updates its estimation via the above-mentioned rule. By continuing this approach, after $d_g$ steps, all of the agents will know the correct values of $\bar{M}$ and $\underset{\raise0.3em\hbox{$\smash{\scriptscriptstyle-}$}}{M}$, where $d_g$ is the diameter of graph $\mathcal{G}$.
	\end{Rem}
	
	\color{black}
	\subsection{\label{subsec.main} Optimization using an uncertain nonlinear multi-agent system in the presence of discontinuous communications}
	
	In this subsection, an algorithm for solving Problem~\ref{problem} is proposed. 
	For this purpose, at first, 
	\color{black} we need to introduce the method of linear active disturbance rejection control (LADRC). A special form of this method is briefly described in Lemma~\ref{lem.ladrc}.
	\begin{Lem} (Linear active disturbance rejection control) \cite{guo2016active}
		\label{lem.ladrc}
		Consider \textcolor[rgb]{0,0,0}{the} dynamic system 
		\begin{align}
			\label{lem.ad.dynamic}
			\dot x(t) = \left( {p(x) + \Delta p(x)} \right) + \omega (t) + \left( {b + \Delta b} \right)u(t),
		\end{align}
		\color{black}
		where $x(t) \in \R$ and $u(t) \in \R$. Also, $b \in \R$ ($\Delta b \in \R$) is a known (an unknown) constant, function $p(x):\R \rightarrow \R$ is known, and $\Delta p(x): \R \rightarrow \R$ and $\omega (t) : \R^{\geq 0} \rightarrow \R$ are unknown functions. 
		\color{black}
		Assume that
		\begin{align}
			\label{ADRC.condition_1}
			&\omega (t) + \Delta p(x) \text{ and } {\dfrac{d}  {dt}\left( {\omega(t)  + \Delta p(x)} \right)} \text{ are bounded,}\\
			\label{ADRC.condition_2}
			&\dfrac{\partial } {\partial x} p(x) \text{ is bounded and } p(x) \text{ is Lipschitz with parameter } L_p,\vspace{1.5mm}\\
			\label{ADRC.condition_3}
			&\textcolor{black}{|\Delta b| \leq \rho_b, \text{ where } \rho_b \text{ is known,} }
		\end{align}
		are simultaneously satisfied. By using \textcolor[rgb]{0,0,0}{the} observer 
		\begin{align}
			\label{ADRC.dynamic}
			\begin{cases}
				\dot{ \hat {x}}(t) = \hat {\bar {x}}(t) + p(\hat x) + {\dfrac{{k_1}}  {\varepsilon} }\left( {x(t) - \hat x(t)} \right) + bu(t) \\
				\dot {\hat{ \bar{ x}}}(t) = {\dfrac{{k_2}} {{\varepsilon ^2}}}\left( {x(t) - \hat x(t)} \right)
			\end{cases},
		\end{align}
		with conditions ${k_1} < 0$ and \textcolor{black}{$\left( {{L_p} + {\dfrac{\rho_b}  {b}}} \right){k_2}<\mathcal{L}$}, where
		\begin{align}
			\nonumber
			\mathcal{L}=\dfrac{2{k_1}{k_2}} { - k_1^2 - {{\left( {{k_2} - 1} \right)}^2} + \sqrt {\left( {k_1^2 + {{\left( {{k_2} - 1} \right)}^2}} \right)\left( {k_1^2 + {{\left( {{k_2} + 1} \right)}^2}} \right)} }
		\end{align}
		and $
		u(t) = {\dfrac{1}{b}}\left( {\alpha \hat x - \hat {\bar x} - p(\hat x)} \right)$ for some $\alpha < 0$,
		then there exists constant $\varepsilon_0 >0 $ such that for each $\varepsilon \in (0,\varepsilon_0)$, $t_{\varepsilon} >0 $ is found where $
		\left| {x(t)} \right| < \Gamma \varepsilon $
		for all $\varepsilon \in (0,\varepsilon_0)$ and $t > t_{\varepsilon}$. Also, $\Gamma >0$  is an $\varepsilon$-independent positive constant. 
	\end{Lem}
	
	Before \textcolor{black}{introducing a solution for Problem~\ref{problem},} we need to show how the assumptions on function $f$ in Lemma~\ref{lem.cortes} are satisfied by Assumption~\ref{as.optimization}\ref{as.convex}. \textcolor[rgb]{0,0,0}{This aim is achieved by the following lemma.}
	\begin{Lem} (Relation between convexity and Lipschitz conditions)
		\label{lem.relation}
		Consider function $f\left( {{x_1},...,{x_n}} \right) = \sum\nolimits_{i = 1}^n {{f_i}({x_i})} $, where ${f_i}({x_i}):\mathbb{R} \to \mathbb{R}$ for $i=1,2,...,n$.
		\begin{enumerate}
			\item $\nabla f$ is strictly monotonic if and only if each $f_i$ is strictly convex for all $i=1,2,...,n$ .
			\item  $f$ is differentiable with locally Lipschitz partial derivatives if $f_i$ is twice differentiable with bounded second derivative for all $i=1,2,...,n$.
		\end{enumerate}
		
	\end{Lem}
	\color{black}
	This lemma is proved in Appendix. 
	\color{black} 
	Now, Theorem~\ref{thm.decentralized}, Lemma~\ref{lem.ladrc} and Lemma~\ref{lem.relation} are applied to introduce \textcolor[rgb]{0,0,0}{a solution for Problem~\ref{problem}} in the following theorem. 
	\begin{Thm} (Distributed optimization by a dynamic multi-agent system and discontinuous communications) 
		\label{thm.main}
		Consider the multi-agent system \eqref{dynamic} meeting Assumption~\ref{as.coeff}. Assume that the aim is to control of this system in such a way that the agents cooperatively solve the distributed optimization problem \eqref{optimization} satisfying Assumption~\ref{as.optimization}. Let $y_i(t)$ be obtained from\textcolor{black}{~\eqref{decentral.dynamic}, in which $\tilde y_i(t)$ and $\tilde \mu_i(t)$ are updated from \eqref{tilde.def}} by considering the triggering condition \eqref{decental.trigger}. If $\hat e_i(t)$ and $\hat{\bar e}_i (t) $ are updated by
		\begin{align}
			\begin{cases}
				\label{e_definition}
				{{{\dot {\hat e}}_i} = {{\hat {\bar e}}_i} + \left( {{p_i}({{\hat e}_i} + {y_i}) - {p_i}({y_i})} \right) + \frac{{{k_{1i}}}}{{{\varepsilon _i}}}\left( {{x_i} - {y_i} - {{\hat e}_i}} \right) + {b_i}{u_i}} \\
				{{{\dot {\hat {\bar e}}}_i} = \frac{{{k_{2i}}}}{{{\varepsilon _i}^2}}\left( {{x_i} - {y_i} - {{\hat e}_i}} \right)}
			\end{cases}; \quad 1 \leq i \leq n,
		\end{align}
		where $k_{1i} \leq 0$ and the control signals are given by
		\begin{align}
			\label{controller}
			{u_i}(t) = \frac{1}{b_i}\left( {\alpha {{\hat e}_i} - {{\hat {\bar e}}_i} - {p_i}({{\hat e}_i} + {y_i}) + {p_i}({y_i})} \right),\quad \alpha<0,
		\end{align}
		for all $1\leq i\leq n$, then $x_i(t)$ converges to the solution of the optimization problem \eqref{optimization} as ${\varepsilon _i} \to 0$. 
	\end{Thm}
	\begin{proof} For \textcolor[rgb]{0,0,0}{the} agent $i$ of system \eqref{dynamic}, define a\textcolor{black}{n} augmented agent with the state vector \textcolor{black}{$\begin{bmatrix} x_i(t) & y_i(t) & h_i(t) \end{bmatrix}{^\top}$}, \textcolor[rgb]{0,0,0}{where} ${h_i} = \begin{bmatrix}\mu_{i_1}&\mu_{i_2}&...&\mu_{i_k}\end{bmatrix}_{i_j \in K(i)}$ \textcolor[rgb]{0,0,0}{with} $K(i)=\left\{k | c_{ki} \neq 0\right\} $. Assume that the state \textcolor[rgb]{0,0,0}{variables} of each augmented agent update according to\textcolor{black}{~\eqref{decentral.dynamic}} and \eqref{decental.trigger}. Also, assume that each agent can share the information of its augmented agent with its neighbors when it communicates with them. According to the Theorem~\ref{thm.decentralized}, it is concluded that $\mathop {\lim }\limits_{t \to \infty } {y_i}(t) = y_i^{\star}$, where ${y^{\star}} = \begin{bmatrix} y_1^{\star}&...&y_n^{\star} \end{bmatrix}^{\top}$ is the solution of optimization problem \eqref{optimization}. By defining $e_i=x_i-y_i$, \eqref{dynamic} and\textcolor{black}{~\eqref{decentral.dynamic}} result in
		\begin{align}
			\label{e_dynamic}
			{\dot e_i} = \left( {{p_i}\left( {{x_i}} \right) + \Delta {p_i}\left( {{x_i}} \right)} \right) + \left( {{b_i} + \Delta {b_i}} \right){u_i} - {g_i}\textcolor{black}{(\tilde y,\tilde \mu )},
		\end{align}
		\textcolor{black}{which is written as}
		\begin{align}
			\label{e_dynamic_1}
			{{\dot e}_i} = \left( {{{\bar p}_i}\left( {{e_i}} \right) + \Delta {p_i}\left( {{e_i} + {y_i}} \right)} \right) + \left( {{p_i}({y_i}) - {g_i}\left( {\tilde y,\tilde \mu } \right)} \right) 
			+ \left( {{b_i} + \Delta {b_i}} \right){u_i}(t),
		\end{align}
		where $
		{\bar p_i}\left( {{e_i}(t)} \right) = {p_i}\left( {{e_i}(t) + {y_i}(t)} \right) - {p_i}({y_i}(t))$. By comparing \eqref{e_dynamic_1} and \eqref{lem.ad.dynamic}, from Lemma~\ref{lem.ladrc} the following fact is deduced. If conditions \eqref{ADRC.condition_1}-\eqref{ADRC.condition_3} \textcolor[rgb]{0,0,0}{are simultaneously satisfied}, by applying the control signal \eqref{controller} there is constant $\varepsilon_{i_0}$ such that for \textcolor[rgb]{0,0,0}{each $\varepsilon_i \in (0,\varepsilon_{i_0})$, $t_{\varepsilon_i}$ yielding in $\left| {{e_i}(t)} \right| < {\Gamma _i}{\varepsilon _i}$ for all $\varepsilon_i \in (0,\varepsilon_{i_0})$ and $t>t_{\varepsilon_i}$ is found. }In this case, $\Gamma_i$ is an $\varepsilon_i$- independent positive constant. This statement is equivalent to that $x_i(t)$ converges to the solution of optimization problem \eqref{optimization} as $\varepsilon_i \rightarrow 0$. In order to complete the proof, at the final step it is shown that conditions \eqref{ADRC.condition_1} and \eqref{ADRC.condition_2} are satisfied. To satisfy \eqref{ADRC.condition_1}, it is necessary to examine the boundedness of two functions $\Delta {p_i}\left( {{e_i} + {y_i}} \right) + \left( {{p_i}({y_i}) - {g_i}\textcolor{black}{\left( {\tilde y,\tilde \mu } \right)}} \right)$ and $\frac{{d\left( {\Delta {p_i}\left( {{e_i} + {y_i}} \right) + \left( {{p_i}({y_i}) - {g_i}\textcolor{black}{\left( {\tilde y,\tilde \mu } \right)}} \right)} \right)}}{{dt}}$. This proposition is deduced by considering Assumption~\ref{as.coeff}\ref{as.bounded} and noting that $\dot y_i(t)$, $y_i(t)$, $\mu_i(t)$, $\tilde{y}_i(t)$, and $\tilde \mu_i(t)$ are bounded. On the other hand, according to Assumption~\ref{as.coeff}\ref{as.lipschitz}, $p_i$ is a Lipschitz function with parameter $L_{p_i}$. Considering the definition of function $\bar{p}_i$, it can be shown that function $\bar{p}_i$ is also a Lipschitz function with parameter $L_{p_i}$. Furthermore, according to Assumption~\ref{as.coeff}\ref{as.bounded}, ${{\partial {p_i}({x_i})} \mathord{\left/
				{\vphantom {{\partial {p_i}({x_i})} {\partial {x_i}}}} \right.
				\kern-\nulldelimiterspace} {\partial {x_i}}}$ is bounded. Hence, ${{\partial {{\bar p}_i}({x_i})} \mathord{\left/
				{\vphantom {{\partial {{\bar p}_i}({x_i})} {\partial {x_i}}}} \right.
				\kern-\nulldelimiterspace} {\partial {x_i}}}$ is bounded, and consequently condition \eqref{ADRC.condition_2} is satisfied.
	\end{proof}
	\textcolor{black}{
		The result of Theorem~\ref{thm.main} on solving Problem~\ref{problem} is summarized in Algorithm~\ref{alg.main}.
		\begin{algorithm}
			\color{black}
			\caption{\color{black}Solving distributed optimization \eqref{optimization} by the uncertain multi-agent system~\eqref{dynamic} with discontinuous communications}
			\label{alg.main}
			\begin{itemize}
				\item\textbf{Initialization}: Set $y_i(0)=x_i(0)$ and $\hat{e}_i(0)=\hat{\bar{e}}_i(0)=0$ for $i=1,...,n$. Also, set $\mu_{i-n}(0) = 0$ for all $i=n+1,...,n+m$. Moreover, for the agent $i$ ($i=1,..,n$), form the augmented state vector $\begin{bmatrix} x_i(t) & y_i(t) & h_i(t) \end{bmatrix}\textcolor{black}{^\top}$, where ${h_i(t)} = \begin{bmatrix}\mu_{i_1}(t)&\mu_{i_2}(t)&...&\mu_{i_k}(t)\end{bmatrix}_{i_j \in K(i)}$ and $K(i)=\left\{k | c_{ki} \neq 0\right\} $. Furthermore, \textcolor{black}{set ${t}{^i_1}=0$} for all $i=1,...,n+m$.
				\item \textbf{Do} 
				\begin{itemize}
					\item Define $\tilde{y}_i(t) = y_i({t}{^i_{k_i}})$ for $i=1,...,n$ ($\tilde{\mu}_{i-n}(t) = \mu_{i-n}({t}^i_{k_i})$ for $i=n+1,...,n+m$), where ${t}^i_{k_i}$ is the last update time of $y_i$ ($\mu_{i-n})$.
					\item Update $y_i(t)$ and $\mu_{l}(t)$ from\textcolor{black}{~\eqref{decentral.dynamic}} by the agent $i$ ($i=1,...,n$) for $l \in K(i)$ (For updating $\mu_{l}(t)$ where $l \in K(i)$, agent $i$ only needs the information received in the last broadcast time).
					\item Update $\hat{e}_i(t)$ and $\hat{\bar{e}}_i(t)$ from~\eqref{e_definition}, and the controller signal $u_i(t)$ from~\eqref{controller} by the agent $i$ ($i=1,...,n$).
					\item Apply the control signal to~\eqref{dynamic}.
				\end{itemize}
				\item \textbf{\textbf{While} condition \eqref{decental.trigger} is not satisfied.}
				\item \textbf{Set} ${t}^i_{k_i+1}=t$ and $k_i \leftarrow k_i+1$, when condition \eqref{decental.trigger} is met by the agent $i$ ($i=1,...,n+m$). \textbf{If} the agent index $i$ satisfies $1 \leq i \leq n$, then agent $i$ should broadcast $y_i(t)$ to its neighbors.  \item \textbf{else} each agent with non-zero corresponding element in the $(i-n)$'s row of $C$ should broadcast $\mu_{i-n}(t)$ to its neighbors which are not involved in $(i-n)$'s equality constraint of optimization problem \eqref{optimization}.
				\item\textbf{Go to} step 2.
			\end{itemize}
		\end{algorithm}
	}
	\color{black}
	\section{\label{sec.application}Application in Resource Allocation \textcolor{black}{and Examples}}
	
	In this section, two resource allocation based sample applications for the \textcolor[rgb]{0,0,0}{results obtained in} the previous section are discussed. Also, related numerical examples are given to show the efficiency of the \textcolor[rgb]{0,0,0}{introduced algorithm}.
	\subsection{\label{sec.allocation1}Case I}
	
	Resource allocation problem is widely raised in different fields such as communication/sensors networks, economical systems, and power grids \cite{zhao2014design}. \textcolor{black}{On the basis of the works \cite{richert2016distributed} and \cite{yi2016initialization}, {solving} some resource allocation problems can be considered as a sample application for the {obtained} results.} According to the obstacles of using centralized methods in resource allocation problems (e.g., low efficiency of centralized methods in complex networks, high communication cost, privacy concerns, and time-delay challenges \cite{yi2016initialization}), proposing fully distributed algorithms seems to be more useful in practice. In recent years, this issue has been widely considered in literature and addressed from different aspects \cite{ghadimi2013multi}. In the modeling of a resource allocation problem \cite{yi2016initialization}, $k$ storing tasks should be done by $l$ agents. Each store-place has a specific capacity and the agents should cooperatively fill the capacity of all store-places. On the other hand, each agent has a specific amount of resources and benefits from any task proportional to the amount of resources sent to that task. The objective in this network system is to maximize the sum of the benefits of the all agents. In this resource allocation system, each agent knows only its local data. This data includes the benefit of sending resources for each task, amount of its resources, and the load of each task. Furthermore, in this system, \textcolor[rgb]{0,0,0}{the} agents share their data with some of the other agents through a \textcolor[rgb]{0,0,0}{communication} graph. This problem can be represented as a distributed optimization problem with $n$ agents ($n\leq k\cdot l$) and $m=k+l$ linear equality constraints as 
	\begin{align}
		\nonumber
		\left\{\begin{matrix*}[l]
			\min\limits_{x \in \R^n} & f(x) = \sum\limits_{i = 1}^l {\sum\limits_{j \in {K_i}}^{} {{f_{i,j}}({x_{i,j}})} } \vspace{1.5mm}  \\
			~{\rm s.t.}  & \left\{\begin{matrix*}[l]\sum\limits_{j \in {K_i}} {{x_{i,j}}}  = {P_i},&i = 1,2,...,l\\\sum\limits_{i|j \in {K_i}} {x_{i,j}} = {T_j},&j = 1,...,k \end{matrix*}\right.
		\end{matrix*}\,\right.,
	\end{align}
	where ${K_i} \subseteq \left\{ {1,2,...,k} \right\}$ is the union of the tasks that the agent $i$ can do, $P_i$ is \textcolor[rgb]{0,0,0}{the nominal power of the agent} $i$, and $T_j$ is task $j$'s storage capacity. \textcolor{black}{Furthermore, $x_{i,j}$ is the amount of {{the works}} of task $j$ that are done by the agent $i$. Also, function $f_{i,j}(.)$ denotes the cost of doing task $j$ by the agent $i$.} Also, it is assumed that the activity of each agent is specified based on the differential equation 
	\begin{align}
		\label{ex.1.dynamic}
		{\dot x_{i,j}} = \left( {{a_{i,j}} + \Delta {a_{i,j}}} \right){x_{i,j}} + ({b_{i,j}} + \Delta {b_{i,j}}){u_{i,j}},
	\end{align}
	for all $1 \leqslant i \leqslant l$ and $j \in {K_i}$, where the real values $a_{i,j}$ and $b_{i,j}$ are known, and $\Delta a_{i,j}$ and $\Delta b_{i,j}$ are unknown bounded constants. The Algorithm~\ref{alg.main} can be used in above-described resource allocation problem. As a sample, consider the following \textcolor{black}{numerical} example.
	
	\begin{Ex} In the considered resource allocation problem, assume the case that $l=2$, $k=2$, and $n=4$. Also, suppose that the corresponding optimization problem is in the form
		\begin{align}
			\nonumber
			\left\{\begin{matrix*}[l]
				\min\limits_{x \in \R^4} & f(x) =  - 5x_{1,1}^2 - 15x_{1,2}^2 - 20x_{2,1}^2 - 10x_{2,2}^2 \vspace{0.5mm}  \\
				~{\rm s.t.}  & \left\{\begin{matrix*}[l]{{x_{1,1}} + {x_{1,2}} = 1},&{{x_{2,1}} + {x_{2,2}} = 1}\\{{x_{1,1}} + {x_{2,1}} = 1},&{{x_{1,2}} + {x_{2,2}} = 1} \end{matrix*}\right.
			\end{matrix*}\,\right..
		\end{align}
		Furthermore, assume that equation $\dot{x}=(A+\Delta A)x+(B+\Delta B)u$ with $A=\begin{bmatrix} -2&-3\\-4&-5 \end{bmatrix}$ and $b=\begin{bmatrix} 2&3\\4&5 \end{bmatrix}$ describes \textcolor[rgb]{0,0,0}{the dynamic model of} the agents. \textcolor{black}{The communication graph for this problem is shown in Figure~\ref{ex1.1}. 
		} 
		\color{black}
		\begin{figure}[hbt!]
			\centering
			\includegraphics[width=55mm]{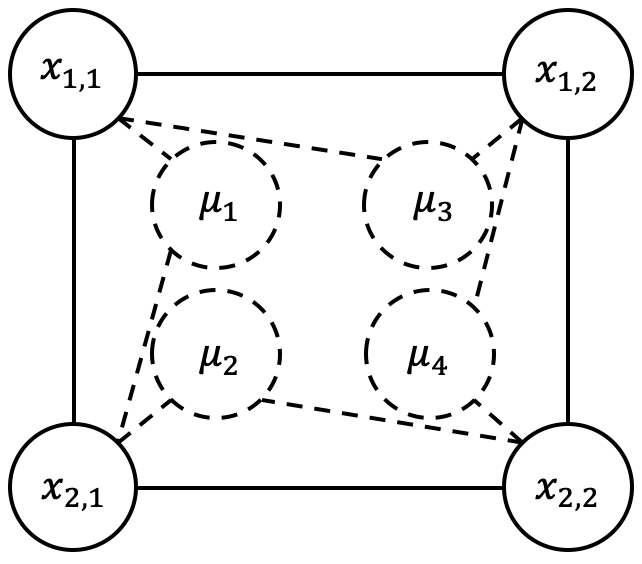}
			\caption{\color{black}{Communication graph of Example 1 (The virtual agents and the their related edges are shown in the dashed form).}}
			\label{ex1.1}
		\end{figure}
		\color{black}
		\\
		\textcolor{black}{Considering two sets of random values for uncertainties $\Delta A$ and $\Delta B$ and applying {Algorithm~\ref{alg.main}} with $\gamma_i > 0$ for all $1 \leq i \leq 4$, simulation results are shown in Figure \ref{ex1}, and compared with the case of using continuous-time approach (i.e., by considering $\gamma_i=0$ for all $1 \leq i \leq 4$). These results confirm that the system state vector converge to the optimum value $\left( {x_{1,1}^{\star},x_{1,2}^{\star},x_{2,1}^{\star},x_{2,2}^{\star}} \right) = \left( {0.7,0.3,0.3,0.7} \right)$ in the both cases. As a remarkable point, the performance of the control system in tracking the optimum solution of the considered optimization problem does not significantly change in the case of using discontinuous communication in comparison with the case of using continuous-time method. Moreover, Figure~\ref{fig.trig} specifies the triggering times of each agent in the case of discontinuous communication. This figure reveals that the number of data broadcasting of each agent decreases as we approach the equilibrium point.
		} 
		\begin{figure}[hbt!]
			\centering
			\includegraphics[width=135mm]{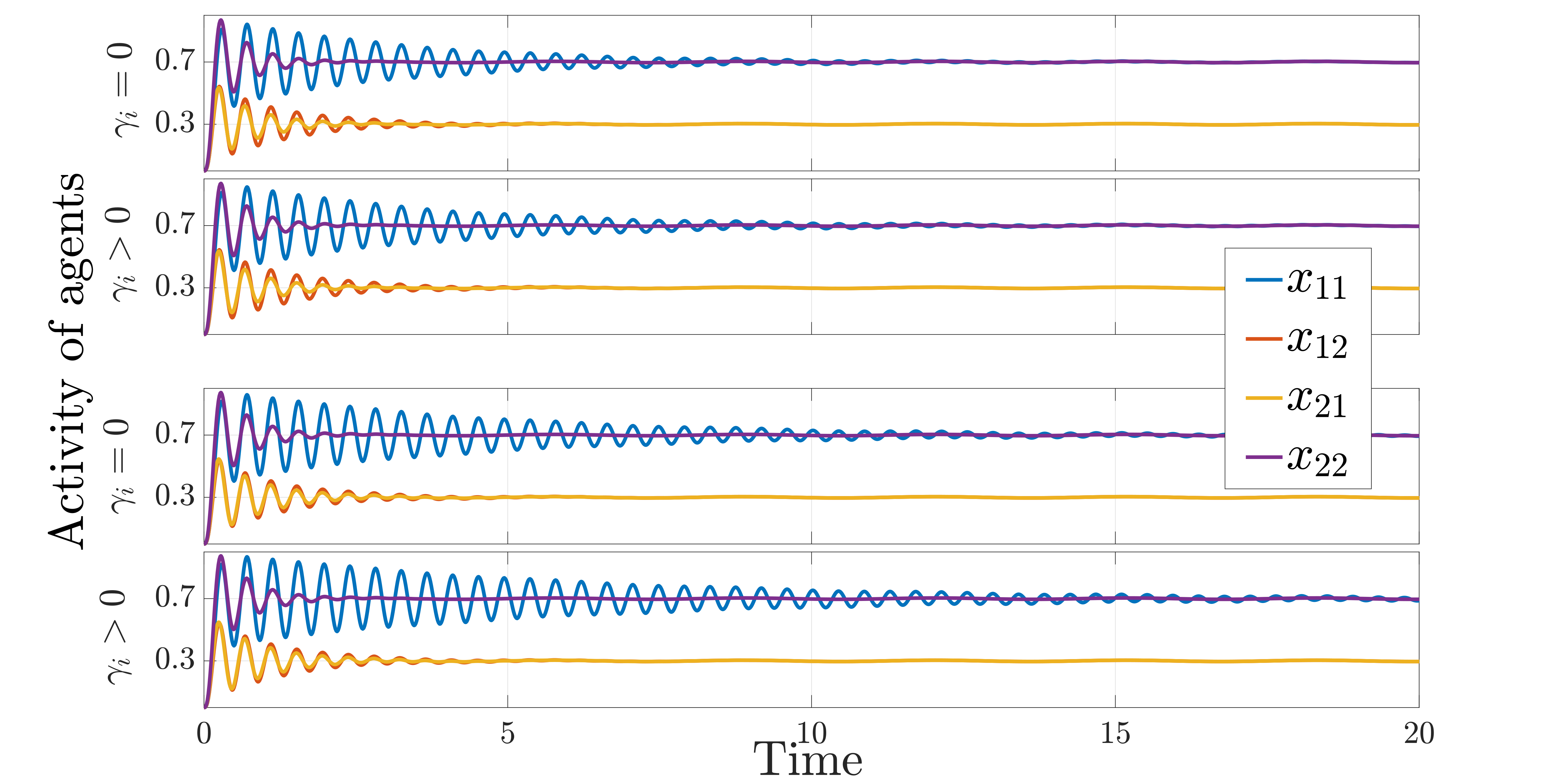}
			\caption{\color{black}{Simulation results of Example 1 for two different sets of random uncertainties in dynamic \textcolor[rgb]{0,0,0}{model} \eqref{ex.1.dynamic}}.}
			\label{ex1}
		\end{figure}
		\begin{figure}[hbt!]
			\centering
			\includegraphics[width=135mm]{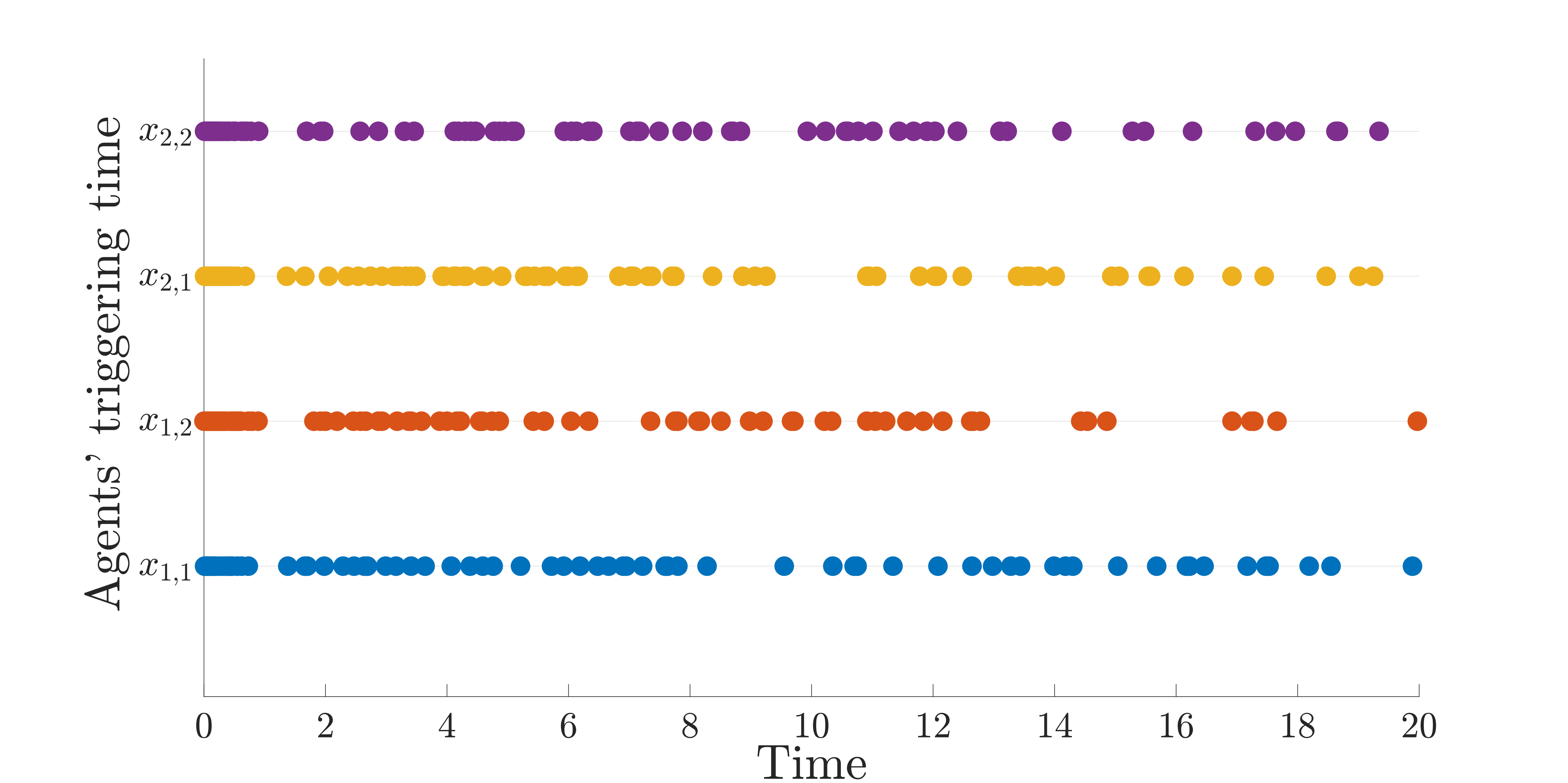}
			\caption{\textcolor{black}{Triggering times of the agents in Example 1.}}
			\label{fig.trig}
		\end{figure}
	\end{Ex}
	\subsection{\label{sec.allocation2}Case II}
	\textcolor[rgb]{0,0,0}{A sample form of the} resource allocation problems is distributed economic dispatch problem (DEDP) in power grids. In DEDP, it is assumed that there exist $n$ controlling areas. Each controlling area, namely controlling area $i$ ($i=1,...,n$), has a local generator supplying power\textcolor[rgb]{0,0,0}{, denoted by} $P_{G_i} \in \R$, and has a local load consuming power\textcolor[rgb]{0,0,0}{, denoted by} $P_{d_i} \in \R$. Also, each local generation has a generation cost (defined by function ${f_i}\left( {{P_{{G_i}}}} \right):\mathbb{R} \to \mathbb{R}$ \textcolor[rgb]{0,0,0}{in controlling area $i$}) and also a transfer cost (defined by function ${g_i}\left( {{P_{{G_i}}} - {P_{{d_i}}}} \right):\mathbb{R} \to \mathbb{R}$ \textcolor[rgb]{0,0,0}{for controlling area $i$}). Generation and cost functions of each controlling area are unknown for the other agents. The total power generating in the network should satisfy constraint $
	\sum\nolimits_{i = 1}^N {{P_{{G_i}}} = \sum\nolimits_{i = 1}^N {{P_{{d_i}}}} }$
	in the steady state. Also, the objective is to optimize the global cost function
	\begin{align}
		\nonumber
		f\left( {{{\left[ {{P_{{G_i}}}} \right]}_{1,2,...,n}}} \right) = \sum\limits_{i = 1}^N {\left( {{f_i}\left( {{P_{{G_i}}}} \right) + {g_i}\left( {{P_{{G_i}}} - {P_{{d_i}}}} \right)} \right)}.
	\end{align}
	It is assumed that the power generation dynamic model for each agent is described by
	\begin{align}
		\label{ex2.dynamic}
		{\dot P_{{G_i}}} = {r_i} \cdot {P_{{G_i}}} + {s_i} \cdot {u_i},\quad 1 \leq i \leq n,
	\end{align}
	where $r_i, s_i \in \R$ for all $1 \leq i \leq n$. In Example 2, which is chosen from~\cite[Example 5.2]{yi2016initialization}, it is shown that \textcolor[rgb]{0,0,0}{Algorithm~\ref{alg.main}} can be used for solving the above-described DEDP.
	\begin{Ex} Consider the DEDP in a 118-bus system with 59 generators. Assume that \textcolor[rgb]{0,0,0}{the} local cost functions of each area \textcolor[rgb]{0,0,0}{\textcolor[rgb]{0,0,0}{are}} in the form
		\begin{align}
			\nonumber
			\begin{cases}
				{{f_i}({P_{{G_i}}}) = {a_i}P_{{G_i}}^2 + {b_i}{P_{{G_i}}} + {c_i}} \\
				{{g_i}({P_{{G_i}}} - {P_{{d_i}}}) = {{a'}_i}{{\left( {{P_{{G_i}}} - {P_{{d_i}}}} \right)}^2} + {{b'}_i}\left( {{P_{{G_i}}} - {P_{{d_i}}}} \right) + {{c'}_i}}
			\end{cases}; \quad 1 \leq i \leq 59,
		\end{align}
		\textcolor{black}{where their uncertain coefficients satisfy ${{a_i},{{a'}_i} \in \left[ {0.0024,0.0679} \right]} \,(\$W^{-2})$, $ {{b_i},{{b'}_i} \in \left[ {8.3391,37.6968} \right]} \,(\$W^{-1}$) and ${{c_i},{{c'}_i} \in \left[ {6.78,74.33}\right]}\, (\$) $.
			Also, assume that ${p_{{d_i}}} \in \left[ {0,300} \right] (W)$.} An undirected ring graph with additional edges (1,4), (15,25), (25,35), (35,45) and (45,50) is considered for sharing and exchanging information between the agents (The above-mentioned specifications for the DEDP have been specified in~\cite[Example 5.2]{yi2016initialization}). In addition, it is assumed that the power of generators is generated according to \eqref{ex2.dynamic} \textcolor[rgb]{0,0,0}{with} the uncertain parameters ${{r_i} \in \left[ {5,10} \right]}$ and ${{s_i} \in \left[ {7,8} \right]}$.
		
		Sample simulation results of using \textcolor[rgb]{0,0,0}{Algorithm~\ref{alg.main}} in solving this DEDP are shown in Figures \ref{ex2.fig1} and \ref{ex2.fig2}. Approaching the power of generators to optimal values of the DEDP, where $\varepsilon_i=0.005$ for $i=1,2,...,n$ in \eqref{e_definition}, is shown in Figure \ref{ex2.fig1}. \begin{figure}[hbt!]
			\centering
			\includegraphics[width=135mm]{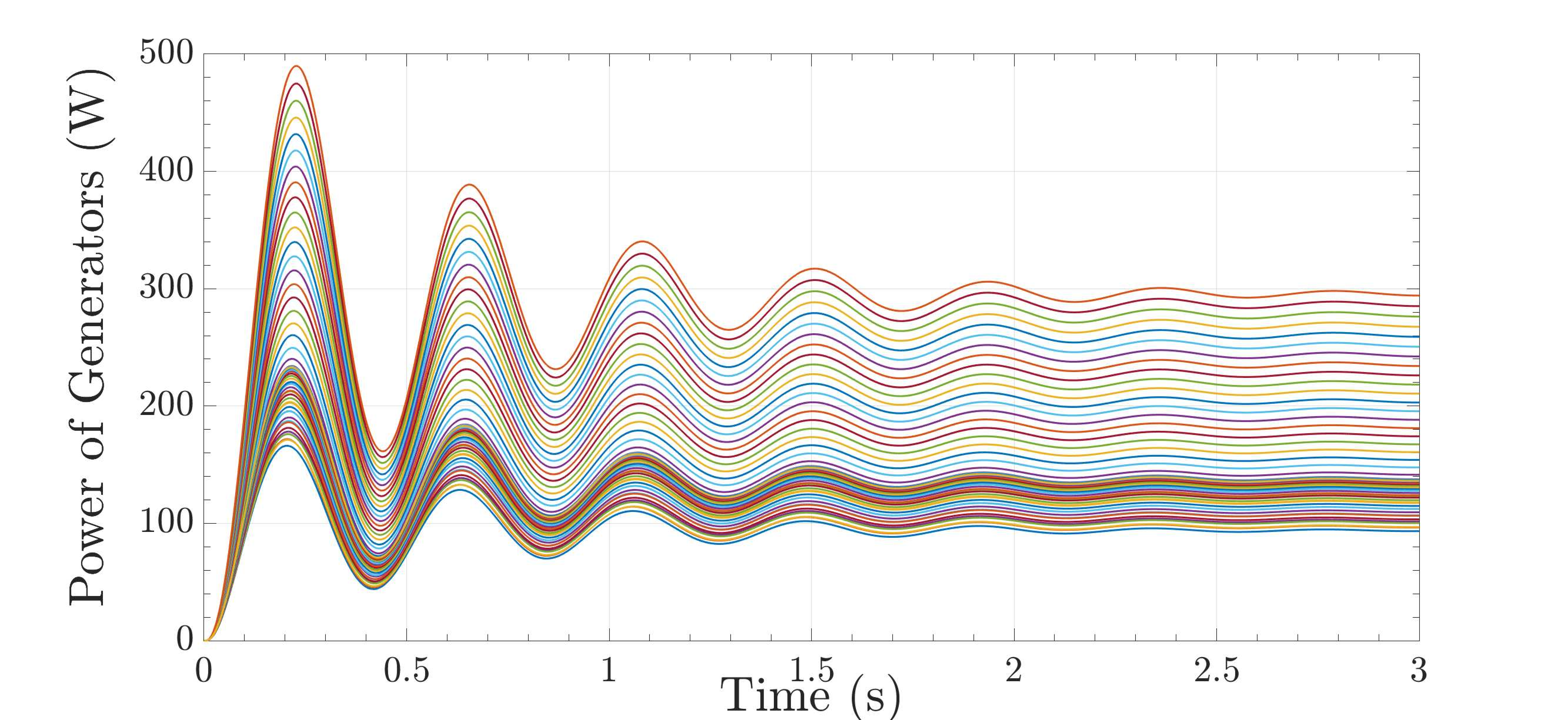}
			\caption{Power of generators in Example 2.}
			\label{ex2.fig1}
		\end{figure}
		\textcolor{black}{Also, Figure~\ref{ex2.fig2} addresses the influence of the free parameters $\varepsilon_i$ of Algorithm~\ref{alg.main} on the distance between the final value of the power of generators and the optimization solution. This figure reveals that a greater value for $\varepsilon_i$ yields in a greater final distance between the optimization solution and the generators' powers.}
		\begin{figure}[hbt!]
			\centering
			\includegraphics[width=135mm]{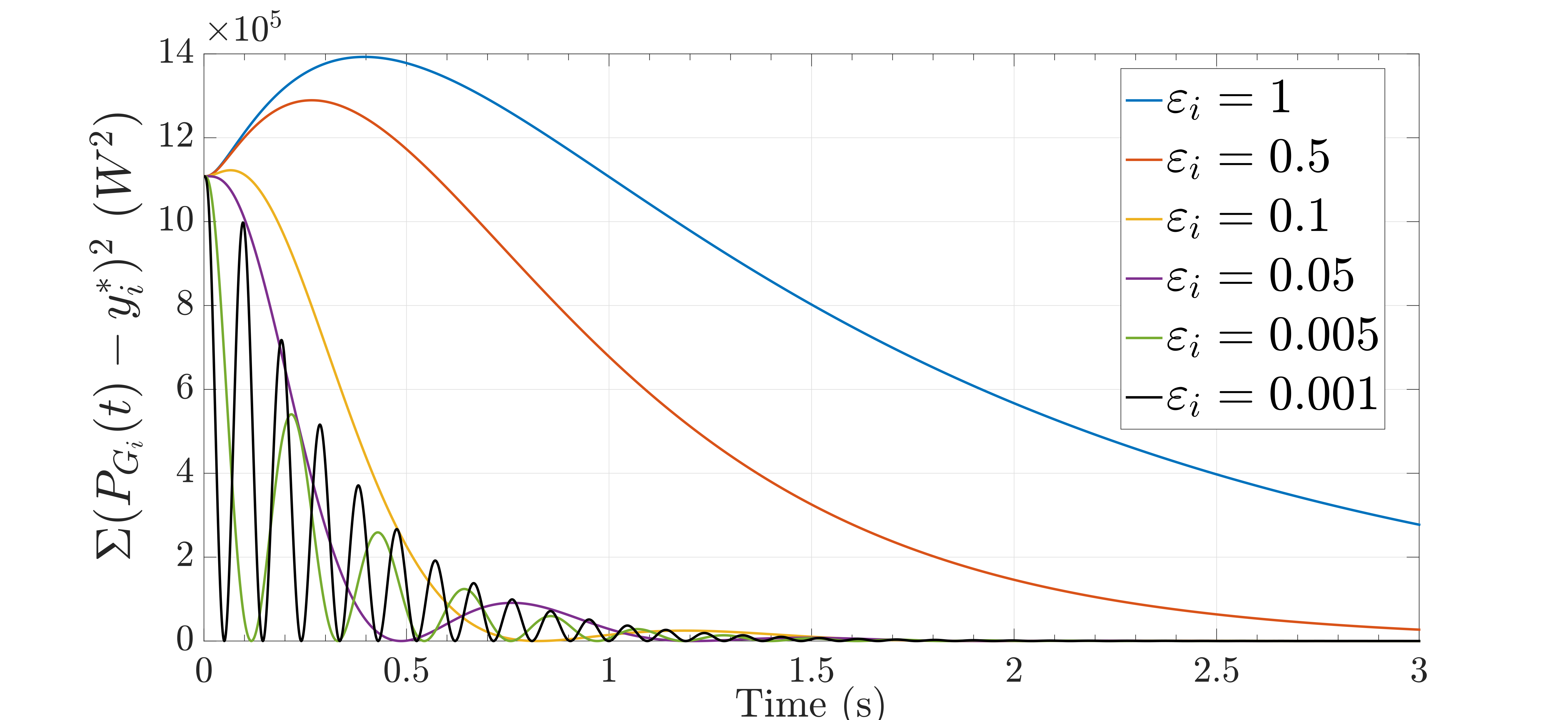}
			\caption{Effect of $\varepsilon_i$ on \textcolor[rgb]{0,0,0}{convergence} speed in Example 2 ($\sum {{{\left( {{P_{{G_i}}}(t) - y_i^{\star}} \right)}^2}} $ versus time where $y_i^{\star}$ denotes the optimal value for $P_{G_i}$). In each simulation of this figure, an equal value is assumed for all $\varepsilon_i$ ($i=1,2,...,n$).}
			\label{ex2.fig2}
		\end{figure}
		To verify the applicability of \textcolor[rgb]{0,0,0}{Algorithm~\ref{alg.main}} in the case of \textcolor[rgb]{0,0,0}{variation} in \textcolor[rgb]{0,0,0}{the} network specifications, in Figure \ref{ex2.fig3} numerical simulation results of using this algorithm with $\varepsilon_i=0.005$ ($1 \leq i \leq n$) are shown in the presence of the following changes: 
		\begin{enumerate}
			\item $ \pm 20\%$ variation in the loads of 18 (randomly chosen) areas at time 2.
			\item $ 0\%-50\%$ change in the value of $a_i$ for 18 (randomly chosen) generators, and $-50\%-0\%$ change in the value of $b_i$ for another 18 (randomly chosen) generators at time 3.
			\item Disconnecting two (randomly selected) buses from the network at time 4.
		\end{enumerate}
		\textcolor{black}{Simulation results in Figure~\ref{ex2.fig3} confirm that Algorithm~\ref{alg.main} is robust against the aforementioned changes, and as the network specifications change, the power of generators converges to the new optimal values.}
		\begin{figure}[hbt!]
			\centering
			\includegraphics[width=135mm]{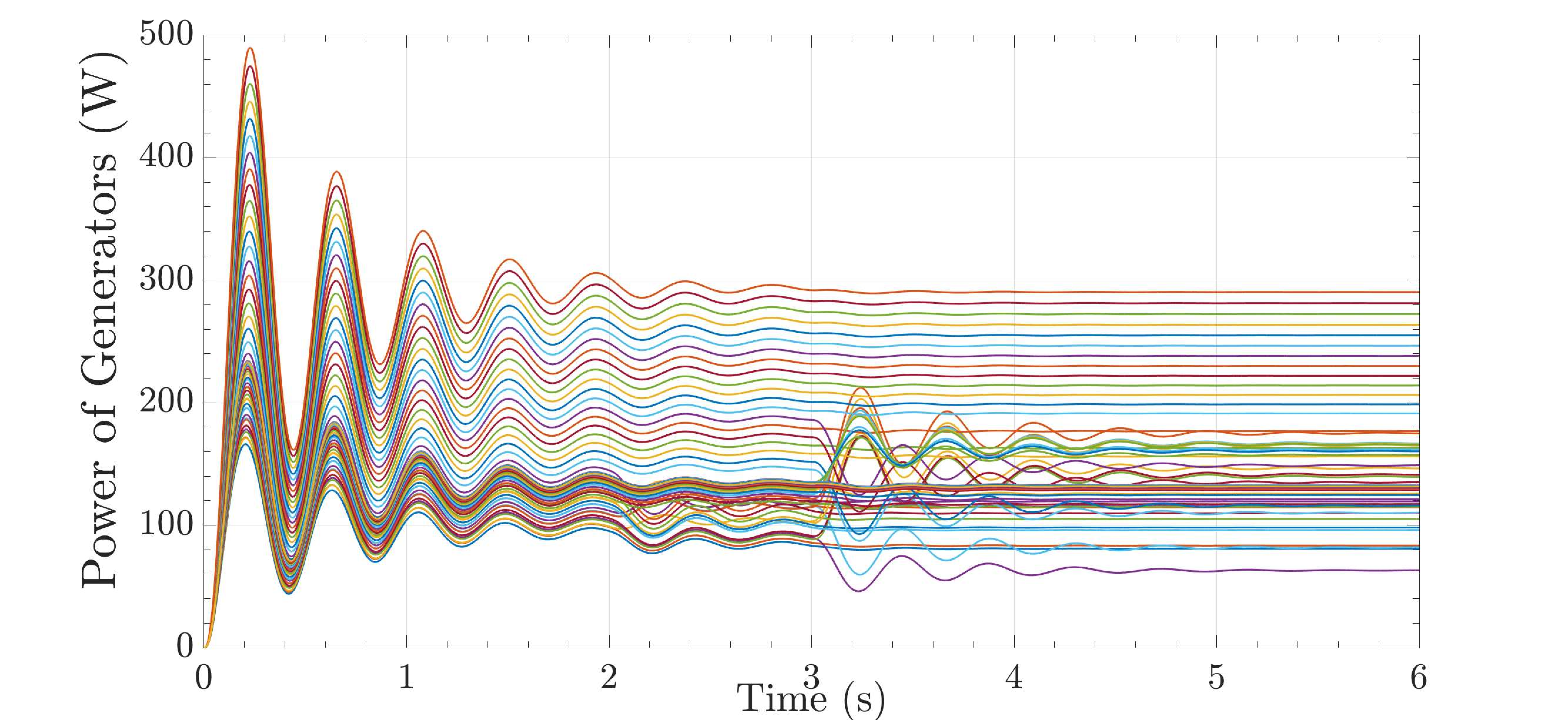}
			\caption{Power of generators in Example 2 in the presence of \textcolor[rgb]{0,0,0}{variations} (1)-(3) in the network specifications.}
			\label{ex2.fig3}
		\end{figure}
	\end{Ex}
	\section{\label{sec:conclusion}Conclusion}
	
	This paper investigated how through event-triggered communications a network of dynamic agents can cooperatively solve a constrained optimization problem. To this end, at first on the basis of Lemma~\ref{lem.cortes} developed in \cite{cortes2018distributed}, Theorem~\ref{thm.centralized} was presented to introduce a centralized event-triggered {mechanism} for solving the considered constrained optimization problem. The result of Theorem~\ref{thm.centralized} was modified to propose a distributed version of the {mechanism} in Theorem~\ref{thm.decentralized}. Furthermore, Theorem~\ref{thm.main} presented an event-triggered based control method for an uncertain dynamic multi-agent system in order to cooperatively solve a constrained optimization problem. \textcolor[rgb]{0,0,0}{Based on this theorem, Algorithm~\ref{alg.main} was proposed to solve Problem~\ref{problem}.} It was verified that the obtained results can be useful in resource allocation problem. 
	
	There are some lines which invite further research works in continuation of this paper. \textcolor{black}{For example, it would be interesting to consider a more general form for the under-study optimization problem with both equality and inequality constraints. Moreover, generalizing the proposed method to solve optimization problems with multi-agent systems which communicate with each other in a directed graph can be considered as another topic for future research.}

	\bibliographystyle{ieeetr}
	\bibliography{ref}
	
	\color{black}
	\appendix \section{Proof of Theorem~\ref{thm.decentralized} and Lemma~\ref{lem.relation}}
	\label{app.proof}
	\textcolor{black}{In this appendix, the proofs of Theorem~\ref{thm.decentralized} and Lemma~\ref{lem.relation} are presented.}
	\begin{proof}[Proof of Theorem~\ref{thm.decentralized}]
		\textcolor{black} {It is worth noting that\textcolor{black}{~\eqref{decentral.dynamic}} can be considered as a distributed relation over the communication graph $\mathcal{G}^a$. This claim is justified by noticing that each (real or virtual) agent for updating its data through\textcolor{black}{~\eqref{decentral.dynamic}} only requires the information shared from itself and its neighbors at the last broadcast time and the information of its involving constraints (Assumption~\ref{as.optimization}(ii)). By considering the local parameters $r_{min_i}$ and definitions of functions $g_i$, it is revealed that the triggering mechanism~\eqref{decental.trigger} is also distributed over graph $\mathcal{G}^a$. To prove the asymptotic convergence of system\textcolor{black}{~\eqref{decentral.dynamic}} with the broadcasting rule \eqref{decental.trigger} to solution of \eqref{lem.opt}, \textcolor{black}{consider the Lyapunov function $v^d(y,\mu ) = \sum\nolimits_{i = 1}^4 {{v_i^d}\left( {y,\mu } \right)} $, where
				\begin{align}
					\nonumber
					\left\{\begin{matrix*}[l]
						{{v_1^d}(y,\mu ) = \frac{1}{2}\sum\limits_{i=1}^n\left\| {g_i\left( {y,\mu } \right)} \right\|_2^2},& {v_2^d}(y,\mu ) = \frac{1}{2}\|Cy-d\|_2^2, \vspace{1mm}\\
						{{v_3^d}(y,\mu ) = \frac{1}{2}\left\| {y - {y^{\star}}} \right\|_2^2},&{v_4^d}(y,\mu ) = \frac{1}{2}\left\| {\mu  - {\mu ^{\star}}} \right\|_2^2.
					\end{matrix*}\right.
				\end{align}
				Note that $v^d(y,\mu )=0$ if and only if $(y,\mu)=(y^{\star},\mu^\star)$. Also, $v^d(y,\mu )>0$, where  $(y,\mu) \neq (y^{\star},\mu^\star)$. Consequently, the convergence proof will be completed if it is verified that the time-derivative of the Lyapunov function candidate $v^d(y,\mu)$ with respect to the dynamic model~\eqref{decentral.dynamic} and the triggering mechanism \eqref{decental.trigger} is negative-definite. By doing some calculations similar to those done in the proof of Theorem~\ref{thm.centralized}, it can be shown that the time-derivative of $v^d(y,\mu )$ satisfies 
				\begin{align}
					\label{revision2equation}
					\dot v^d(t) \leq \frac{{2 + \bar M}}{\kappa }\left\| {{e_y(t)}} \right\|_2^2 + \frac{3}{{2\kappa }}\left\| {{e_\mu(t) }} \right\|_2^2 - h\left( {\tilde y(t),\tilde \mu (t) }\right),
				\end{align}
				where $e_y(t)\coloneqq\begin{bmatrix}e_{y_1}(t)&...&e_{y_n}(t)\end{bmatrix}^\top$, $e_{\mu}(t)\coloneqq\begin{bmatrix}e_{\mu_1}(t)&...&e_{\mu_m}(t)\end{bmatrix}^\top$, and $e_{y_i}$ and $e_{\mu_j}$ were defined in the statement of Theorem~\ref{thm.decentralized}. Moreover, the function $h:\R^{n+m} \rightarrow \R$ is in the form~\eqref{eq.h_definition}. On the other hand}
			\begin{subequations}
				\label{kappa.derivation}
				\begin{align}
					\dfrac{\underset{\raise0.3em\hbox{$\smash{\scriptscriptstyle-}$}}{M}^2}{2(5+3\bar{M})(\bar{M}+2)} &\leq \left(\dfrac{\kappa}{\bar{M}+2}\right)\left(\underset{\raise0.3em\hbox{$\smash{\scriptscriptstyle-}$}}{M}-\dfrac{5+3\bar{M}}{2}\kappa\right)\\
					\dfrac{1}{12}&\leq \left(\dfrac{2\kappa}{3}\right)(1-2\kappa)
				\end{align}
			\end{subequations}
			hold for all $\kappa$ satisfying \eqref{def.kappa} (The inequalities \eqref{kappa.derivation} can be easily proved by finding the values of $\kappa$ which minimize the functions appeared in the right hand sides of \eqref{kappa.derivation} through differentiating these functions with respect to $\kappa$). If $\gamma_i$ satisfies \eqref{gamma_definition}, from \eqref{kappa.derivation} it is found that}
		
		\color{black} 
		\begin{align}
			\label{gamma_definition_2}
			{ \gamma_i^2 } < \min\left\{\left(\dfrac{2\kappa}{3}\right)(1-2\kappa),\left(\dfrac{\kappa}{\bar{M}+2}\right)\left(\underset{\raise0.3em\hbox{$\smash{\scriptscriptstyle-}$}}{M}-\dfrac{5+3\bar{M}}{2}\kappa\right) \right\}
		\end{align}
		for all $1 \leq i \leq n+m$ and all $\kappa$ satisfying~\eqref{def.kappa}. Furthermore, if condition \eqref{decental.trigger} is not met, then from \eqref{gamma_definition_2} \color{black} it is found that
		
		\begin{multline}
			\label{eq.ineq}
			\dfrac{{2 + \bar M}}{\kappa }\left\| {{e_y(t)}} \right\|_2^2 + \dfrac{3}{{2\kappa }}\left\| {{e_\mu(t) }} \right\|_2^2 = \dfrac{{2 + \bar M}}{\kappa }\sum\limits_{i=1}^n e_{y_i}^2(t) + \dfrac{3}{{2\kappa }}\sum\limits_{i=n+1}^{n+m} e_{\mu_{i-n}}^2(t)\vspace{1.5mm}\\ 
			\leq \dfrac{{2 + \bar M}}{\kappa } \sum\limits_{i = 1}^n \gamma_i^2 g_i^2(\tilde y,\tilde \mu)+ \dfrac{3}{{2\kappa }}\sum\limits_{i=n+1}^{n+m} \gamma_i^2 (c_{i-n}^\top\tilde y(t)-d_{i-n})^2
		\end{multline}
		\color{black}
		is valid for all $t \geq 0$. \textcolor{black}{
			According to the inequality~\eqref{gamma_definition_2} and the definition of function $h$ in \eqref{eq.h_definition}, it is concluded that 
			\begin{align}
				\label{eq.ineq2}
				\dfrac{{2 + \bar M}}{\kappa } \sum\limits_{i = 1}^n \gamma_i^2 g_i^2(\tilde y,\tilde \mu)+ \dfrac{3}{{2\kappa }}\sum\limits_{i=n+1}^{n+m} \gamma_i^2 (c_{i-n}^\top\tilde y(t)-d_{i-n})^2{\leq} h(\tilde y(t),\tilde \mu(t) ).
			\end{align}
			From \eqref{eq.ineq} and \eqref{eq.ineq2}, it is obtained that
			\begin{align}
				\label{eq.ineq20}
				\dfrac{{2 + \bar M}}{\kappa }\left\| {{e_y(t)}} \right\|_2^2 + \dfrac{3}{{2\kappa }}\left\| {{e_\mu(t) }} \right\|_2^2 \leq h(\tilde y(t),\tilde \mu(t) ),
			\end{align}
			Considering the strictness of the inequality~\eqref{gamma_definition_2}, the "$\leq$" symbol in \eqref{eq.ineq2} is replaced by an equality one if and only if $g_i(\tilde y , \tilde \mu)=0$ and $c_{i-n}^\top\tilde y(t)=d_{i-n}$ for all $i$, or equivalently $(y,\mu)=(y^{\star},\mu^\star)$. According to this point and note that the left and right hand sides of \eqref{eq.ineq} are equal where $(y,\mu)=(y^{\star},\mu^\star)$, the "$\leq$" symbol in \eqref{eq.ineq20} is also replaced by an equality one if and only if $(y,\mu)=(y^{\star},\mu^\star)$. Noticing this fact and \eqref{revision2equation}, it is found that the time-derivative of the Lyapunov function $v^d(y,\mu)$ is less than or equal to zero with respect to dynamic model~\eqref{decentral.dynamic} and the broadcasting rule \eqref{decental.trigger} for all $t \geq 0$. Also, it is equal to zero if and only if $(y,\mu)=(y^{\star},\mu^\star)$. As a result, the dynamic model~\eqref{decentral.dynamic} in the presence of the triggering mechanism \eqref{decental.trigger} asymptotically converges \textcolor{black}{to $(y^\star, \mu^\star)$, as the solution of} the optimization problem \eqref{lem.opt}.}

		\textcolor{black}{In the rest of the proof, it is shown that the proposed event-triggered mechanism does not exhibit Zeno behavior. To this aim, firstly Claims 1-3, stated in Subsection \ref{subsec.event}, are proved. Then, on the basis of these claims a proof by contradiction is presented to verify the non-existence of the Zeno behavior.}
		
		\textcolor{black}{\textit{Proof of Claim 1}: According to \eqref{tilde.def}, ${g_i}(\tilde y(t),\tilde \mu (t))$ for $i=1,...,n$ and $c_{(i-n)j}\,\tilde y_j(t) - d_{i-n}$ for $i=n+1,...,n+m$ are constant in the time interval $\left[ {t_{k_i-1}^i,t_{k_i}^i} \right)$, provided that the agent $i$ does not receive any data from its neighbors in the aforementioned time interval. Thus, in such a case according to~\eqref{decentral.dynamic}, it is obtained that}
		\begin{align}
			\nonumber
			\left \{\begin{matrix*}[l]
				{y_i}(t) = \left( {t - t_{k_i - 1}^i} \right){g_i}\left( {\tilde y(t),\tilde \mu (t)} \right), \hspace{13.5mm} 1 \leq i \leq n \vspace{1mm}\\
				\mu_{i-n} (t) = \left( {t - t_{k_i - 1}^i} \right) c_{(i-n)j}\,\tilde y_j(t) - d_{i-n},\, 1 \leq i-n \leq m 
			\end{matrix*}\right. 
		\end{align}
		\textcolor{black}{where $t \in \left[ {t_{k_i-1}^i,t_{k_i}^i} \right)$. In this case, from the second part of the triggering condition \eqref{decental.trigger}, it is revealed that $t_{k_i}^i - t_{k_i - 1}^i$ should be greater than $\gamma_i$, or equivalently $t_{k_i}^i >  t_{k_i - 1}^i + \gamma_i$.}
		
		\textcolor{black}{\textit{Proof of Claim 2}: Assume that all the agents broadcast their information at time $t_a$, and the aim is to investigate the first time after $t_a$ in which an agent broadcasts its information. The first broadcast after time $t_a$ cannot be due to enabling the triggering condition~\eqref{decental.trigger}($i$) because after the time $t_a$ and before the next broadcast time $r_i(t)=0$ for $i=1,...,n+m$. On the other hand, according to \textit{Claim 1}, agent $j$ ($1\leq j \leq n+m$) will broadcast its information after the passage of time $\gamma_j$ if it does not receive data from any of its neighbors. Considering this point, the triggering condition~\eqref{decental.trigger}($ii$), and condition~\eqref{gamma_definition}, the next broadcast after time $t_a$ will not occur sooner that the time $t_a + \underset{i}{\min}\, r_{min_i}$.}
		
		\textcolor{black}{\textit{Proof of Claim 3}: In the case of exhibiting the Zeno behavior, there are infinite numbers of broadcasts in a finite time. Therefore, in this case there exists a finite interval time with a length smaller than $\underset{i}{\min}\, r_{min_i}$, such that the system has infinite numbers of broadcasts in this interval. Consider a subinterval in this time interval in which there are infinite numbers of broadcasts and the time between successive broadcasts is less than $\epsilon$. If $\epsilon$ meets condition \eqref{newinequality1}, there are at least $n +m+ 1$ broadcasts in the considered subinterval. Hence, at least one agent broadcasts its information more than once in this subinterval. Assume that the agent $i$ broadcasts the information of its state variable more than once in the considered subinterval and $t_{k_i-1}^i$ denotes its last broadcast time. According to the fact that $\epsilon < \min r_{min_i}$, it can be concluded that all the neighbors of this agent synchronously broadcast their information at the time $t_{k_i-1}^i$ due to the triggering condition~\eqref{decental.trigger}($i$). Repeating this argument to the second layer of neighbors results that they also broadcast their  information at $t_{k_i-1}^i$. Continuing this argument and considering the fact that the graph $\mathcal{G}^a$ is connected yield in that all the agents synchronously broadcast their information at time $t_{k_i-1}^i$.}

		\textcolor{black}{Now, by considering \textit{Claims 2} and \textit{3}, it is shown that the Zeno behavior is avoided. If the Zeno behavior occurs, according to \textit{Claim 3} there exists a time interval with infinite numbers of broadcasts such that the time between successive broadcasts is less than $\epsilon$ and all the agents synchronously broadcast in one of the broadcast times in this interval. Assume that all the agents broadcast at time $t_a$ in this time interval. Consequently, the next broadcast occurs before time $t_a+\epsilon$. But, it is in contradiction with Claim 2, which emphasizes that the next broadcast after the synchronous broadcast time $t_a$ will occur not sooner than the time $t_a+\underset{i}{\min}\, r_{min_i}$, because $\epsilon <(n+m+1)^{-1} \underset{i}{\min}\, r_{min_i} <\underset{i}{\min}\, r_{min_i}$. Such a contraction means that the Zeno behavior is avoided.}
	\end{proof}
	\begin{proof}[Proof of Lemma~\ref{lem.relation}]
		To proof the claim of Lemma~\ref{lem.relation}, \color{black} we firstly prove that function $g$ is strictly convex if and only if $\nabla g$ is strictly monotonic. It is a famous fact that the differentiable $g$ is strictly convex if and only if condition 
		\begin{align}
			\label{lem.convex}
			g(y) > g(x) + {\left( {\nabla g(x)} \right)^{\top}}(y - x)
		\end{align}
		holds. Condition \eqref{lem.convex} yields in
		${\left( {\nabla g(x) - \nabla g(y)} \right)^{\top}}(y - x) > 0,$
		which means strict monotony of $\nabla g$. On the other hand, if $\nabla g$ is strictly monotonic, then by defining $h(t) \triangleq g\left( {x + t\left( {y - x} \right)} \right)$ for given $x$ and $y$, it is found that
		\begin{align}
			\nonumber
			{h^\prime }(t) = {\left( {\nabla g\left( {x + t\left( {y - x} \right)} \right)} \right)^{\top}}\left( {y - x} \right).
		\end{align}
		From the fact that $\nabla g$ is strictly monotonic, it is deduced that ${h^\prime }(t) > {h^\prime }(0)$. Hence,
		\begin{align}
			\nonumber
			g(y) = h(1) = h(0) + \int\limits_0^1 {{h^\prime }} (t)dt > h(0) + {h^\prime }(0) 
			= g(x) + \nabla g{(x)^{\top}}\left( {y - x} \right) 
		\end{align}
		which yields in strict convexity of $g$. The above-mentioned result reveals that $f_i$ is strictly convex if and only if $\nabla f_i$ is strictly monotonic. Now, to prove the first statement of this lemma, it is sufficient to show that ${\nabla _{\left[ {{x_1}{\text{ }}{x_2}...{x_n}} \right]}}f = \begin{bmatrix}{\nabla _{{x_1}}}{f_1}&{\nabla _{{x_2}}}{f_2}&...&{\nabla _{{x_n}}}{f_n}\end{bmatrix}^{\top}$ is strictly monotonic if and only if $\nabla_{x_i} f_i$ is strictly monotonic for all $1 \leq i \leq n$. Since $\left\langle {\nabla f(x) - \nabla f(x),x - y} \right\rangle  = \sum\nolimits_{i = 1}^n {\left\langle {\nabla {f_i}({x_i}) - \nabla {f_i}({x_i}),{x_i} - {y_i}} \right\rangle }$, from Definition~\ref{def.monotone} it is found that $\nabla f$ is strictly monotonic if and only if $\nabla f_i$ is strictly monotonic for all $1 \leq i \leq n$. 
		
		To prove the second statement of the lemma, assume that the second derivative of $f_i$ is bounded for all $1 \leq i \leq n$ and $\nabla^2 f_i\leq L_i$. This yields in ${\left\| {\nabla {f_i}(x) - \nabla {f_i}(y)} \right\|_i} \leqslant {L_i}\left\| {x - y} \right\|$ which means that \textcolor[rgb]{0,0,0}{the} partial derivatives of $f$ \textcolor[rgb]{0,0,0}{are} Lipschitz.
	\end{proof}

\end{document}

%% file: style.tex
\makeatletter
\def\@settitle{\begin{center}%
		\baselineskip14\p@\relax
		\normalfont\LARGE\scshape\bfseries
		\@title
	\end{center}%
}
\makeatother

\makeatletter

\def\subsection{\@startsection{subsection}{2}%
	\z@{.5\linespacing\@plus.7\linespacing}{.5\linespacing}%
	{\normalfont\large\bfseries}}

\def\subsubsection{\@startsection{subsubsection}{3}%
	\z@{.5\linespacing\@plus.7\linespacing}{.5\linespacing}%
	{\normalfont\itshape}}

\usepackage[colorlinks	= true,
			raiselinks	= true,
			linkcolor	= MidnightBlue,
			citecolor	= Mahogany,
			urlcolor	= ForestGreen,
			pdfauthor	= {Mohammad Saeed Sarafraz},
			pdftitle	= {},
			pdfkeywords	= {},   
			pdfsubject	= {},
			plainpages	= false]{hyperref}
\usepackage[usenames, dvipsnames]{color}
\usepackage{dsfont,amssymb,amsmath,subfig, graphicx,fancyhdr,mdframed}
\usepackage{amsfonts,dsfont,mathtools, mathrsfs,amsthm,wrapfig} 
\usepackage[]{siunitx}
\usepackage{ragged2e}
\usepackage{enumitem}

\allowdisplaybreaks
\usepackage{algorithm}
\usepackage[noend]{algpseudocode}

\allowdisplaybreaks
\date{\today}

\patchcmd\@setauthors
  {\MakeUppercase{\authors}}
  {\authors}
  {}{}
  
\newtheorem{Thm}{Theorem}[section]

\newtheorem{Lem}[Thm]{Lemma}

\newtheorem{As}[Thm]{Assumption}

\newtheorem{Prob}[Thm]{Problem}
\newtheorem{Def}[Thm]{Definition} 
\newtheorem{Rem}[Thm]{Remark}
\theoremstyle{remark}
\newtheorem{Ex}{Example}

%% file: command.tex
\newcommand{\R}{\mathbb{R}}

\newcommand{\N}{\mathbb{N}}

\def\be{\begin{equation}}
\def\ee{\end{equation}}